\documentclass[12 pt, reqno]{amsart}
\usepackage{graphicx} 
\usepackage{amssymb}
\usepackage{braket}
\usepackage[dvipsnames]{xcolor}
\usepackage{mathrsfs}
\usepackage{mathtools}
\usepackage{hyperref}
\usepackage{cancel}
\usepackage{amsmath}
\usepackage{amsthm}
\usepackage{amsfonts}
\usepackage{bbold}
\usepackage{setspace}
\usepackage{parskip}
\usepackage{enumitem}
\usepackage{centernot}
\usepackage{slashed}
\usepackage{nicematrix}
\usepackage [T2A]{fontenc}

\newtheorem{theorem}{Theorem}[section]
\newtheorem{lemma}[theorem]{Lemma}
\newtheorem{corollary}[theorem]{Corollary}

\theoremstyle{definition}  
\newtheorem{definition} [theorem] {Definition} 

\newtheorem{example} [theorem] {Example}

\newcommand{\F}{{\mathbb{F}}}

\newcommand{\br}{\rangle}
\renewcommand{\ket}{\rangle}
\newcommand{\bl}{\langle}
\renewcommand{\bra}{\langle}
\newcommand{\G}{\Gamma}
\newcommand{\K}{\Delta_+}

\usepackage[backend=biber]{biblatex}
\title{Supersymmetry Breaking in Graph Quantum Mechanics}
\author{Bek Herz}
\email {erherz15@gmail.com}
\author{Ivan Contreras}
\email{icontreraspalacios@amherst.edu}

\begin{document}

\begin{abstract}
In this paper, we develop the groundwork for a graph theoretic toy model of supersymmetric quantum mechanics. Using discrete Witten--Morse theory, we demonstrate that finite graphs have a natural supersymmetric structure and use this structure to incorporate supersymmetry into an existing model of graph quantum mechanics. We prove that although key characteristics of continuum supersymmetric systems are preserved on finite unweighted graphs, supersymmetry cannot be spontaneously broken. Finally, we prove new results about the behavior of supersymmetric graph quantum systems under edge rewiring. 
\end{abstract}
\maketitle
\section{Introduction} \label{introduction}
Discrete methods in mathematical physics, including spectral graph theory, discrete Morse theory, and simplicial complexes, have proven to be instrumental in understanding the intricate details of quantum field theory. In this manuscript, we use these methods to develop the groundwork for a graph theoretic toy model of supersymmetric quantum mechanics.

There has been some recent work on the supersymmetric nature of graphs, including a graph theoretic interpretation of Morse inequalities via supersymmetric quantum mechanics on graphs \cite{Contreras_Xu_2019}, building on Forman's discussion of combinatorial Morse theory \cite{Forman_2002,Forman_1998_2,Forman_1998}, and, independently, studies in spectral graph theory by Ogurisu \cite{ Ogurisu_2003, Ogurisu_2002}, Requardt,  \cite{ Requardt_2002, Requardt_2005}, and Post, \cite{Post_2009,Post_2008}. In our work, we use finite graphs to model supersymmetric quantum systems, inspired by Witten's approach to supersymmetric quantum mechanics and Morse theory \cite{Witten_1982}.

In a supersymmetric (SUSY) physical theory, for each kind of boson there is an associated kind of fermion of equal mass, called a superpartner. Hermitian symmetry operators called supercharges map between bosonic and fermionic states. Supersymmetric Standard Model extensions pose elegant resolutions to many major open questions in high energy physics, including the hierarchy problem and quantum gravity. Imposing supersymmetric constraints can provide just enough additional structure to make an unsolvable quantum field theory problem a solvable case study. Additionally, insights from supersymmetry have driven significant developments in mathematical physics, including supergeometry and a physical interpretation of Morse theory.

When studying supersymmetric quantum mechanics, we typically consider the configuration space of states to be a differentiable finite dimensional manifold. Here, we replace the continuous manifold with finite graphs. We believe that incorporating supersymmetry into a graph quantum mechanical model is a natural next step in the study of quantum theories on graphs.

We adopt a dual approach (Figure \ref{fig:SUSY Approach}). On the one hand, we can seek to extract physical results from combinatorial Morse theory on graphs by translating the framework back into the language of supersymmetry (red arrow). On the other hand, we can approach graph supersymmetry by directly discretizing supersymmetric quantum theories (blue arrow). We first remind the reader of the established method of discretization of (non-supersymmetric) quantum mechanics on graphs (Section \ref{graph QM}), the key features of  supersymmetric systems (Section   \ref{SUSY}), and the fundamentals of discrete Morse theory (Section \ref{discrete Morse}). Then, in Section \ref{Results}, we combine these three concepts, presenting new results of supersymmetry on combinatorial graphs. We conclude in Section \ref{conclusions} with a brief summary and directions for future research in this area, including combinatorial QFT on graphs \cite{contreras2024combinatorial}.

\begin{figure}[ht]
    \centering
    \includegraphics[width=0.8\linewidth]{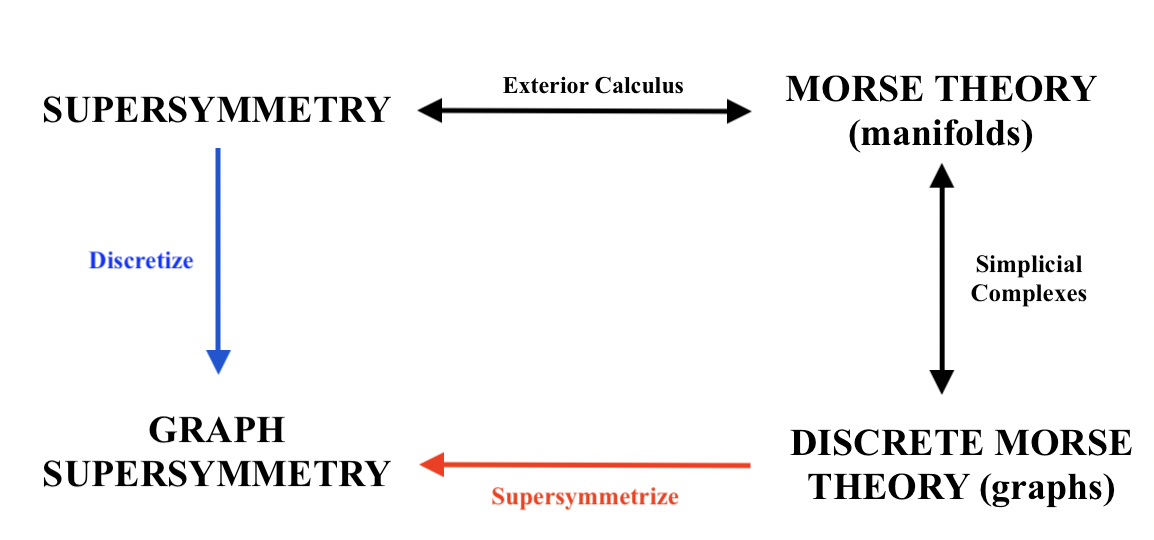}
    \caption{\small{The path to graph supersymmetry: a dual approach.}}
    \label{fig:SUSY Approach}
\end{figure}

\subsection*{Acknowledgments} I.C. and B.H. thank Kannan Jagannathan for co-supervising B.H.'s senior thesis, from which the main results originated. Excerpts from B.H.'s thesis are reproduced here. I.C. and B.H. thank Joe Kraisler for insightful discussions of spectral graph theory and Regan Williams for helpful revisions of the manuscript. I.C. also thanks Pavel Mnev, Santosh Kandel and Konstantin Wernli for useful discussions.

\section{Overview of Graph Quantum Mechanics} \label{graph QM}

We remind the reader of an interpretation of graphs as discrete position-basis quantum mechanical systems.  We may picture the vertices of graph as a discrete sampling of position space. A metric graph $\Gamma_\omega$ comes with an associated function $$\omega: V\times V \rightarrow \mathbb{R},$$ satisfying $$\omega(v_i,v_j)=\omega(v_j,v_i)\;\;\; \text{and}\;\;\; \omega \geq 0,$$
for all adjacent vertices $v_i, v_j \in V(\G_\omega)$, which applies a weight to each edge. Such metrics encode a ``physical distance'' on the graph. In this paper, we focus only on the case of \textit{combinatorial graphs,} graphs with $\omega(v_i,v_j)=1$ for all adjacent vertices $v_i,$ $v_j$. 

In the continuum, a quantum state takes some assigned complex value at each point in space. A \textit{graph} quantum state takes a complex value at each vertex and each edge. A \textit{vertex state} $|\psi\rangle _v$ assigns a complex number to each vertex $v\in V(\Gamma)$. An\textit{ edge state }$|\psi\rangle _e$ assigns a complex number to each edge $e\in E(\Gamma)$. The space of complex valued vertex states $\mathbb{C}^V$ of a graph has dimension $|V|,$ equal to the number of vertices in the graph and the space of complex valued edge states  $\mathbb{C}^E$ has dimension $|E|,$ equal to the number of edges in the graph. 
Thus we can identify the vertex and edge sets with the canonical bases for graph quantum states. For a vertex state operator $\hat O: \mathbb{C}^V\rightarrow\mathbb{C}^V$, we can express the matrix elements:   \begin{equation*}
  [\hat O]_{uv}= \bra u| \hat O |v\ket,
\end{equation*} and likewise for an edge state operator\footnote{When unclear if an operator acts on vertex states or edge states, we distinguish with a ``$+$'' subscript for vertex state operators and a ``$-$'' subscript for edge state operators.}.

Let $\Gamma$ be a graph with vertices $v_1, v_2, ..., v_n$. The adjacency matrix of $\Gamma$ is the $n\times n-$matrix $A: \mathbb{C}^V\rightarrow \mathbb{C}^V,$ with $(i,j)$ entries  \begin{equation} \label{eq:adjacency mtx} [A]_{ij}= \begin{dcases} 1, &\text{if}\;\ v_iv_j\in E(\Gamma) \\ 0, &\text{otherwise}\;\;\;\;\;\;\;\;\;\; .  \end{dcases}  \end{equation}
The adjacency matrix acts on a vertex basis state as \begin{equation}
    A|v_i\ket=\sum_{v_j\in V}[A]_{v_jv_i}|v_j\ket. \end{equation} If a bijection from the adjacency matrix of one graph to the adjacency matrix of another exists, the graphs are isomorphic. The adjacency matrix also helps us understand walks on a graph and is essential to defining a graph analog of the path integral. 
    
\begin{theorem} \label{Th:number_of_walks} Let $\Gamma$ be a graph with vertices $v_1,...,v_n$ and let $A$ be its adjacency matrix. For any positive integer $k$, the matrix element $[A^k]_{ij}=\bra v_i|A^k|v_j\ket$  is the number of walks from $v_i$ to $v_j$ of length $k$ \end{theorem} The well-known proof follows by induction on $k$. See  \cite{Harris_Hirst_Mossinghoff_2008}.

The incidence matrix $I$ catalogs which edges touch which vertices and the orientation of each edge. For a graph $\Gamma$ with $m$ vertices and $n$ edges. Then $I(\Gamma) $ is an $m\times n-$matrix with elements
\begin{equation}
    [I]_{ij}=\begin{cases} 1, & \text{edge $j$ ends at vertex $i$} \\ -1, & \text{edge $j$ starts at vertex $i$} \\ 0, & \text{otherwise} \end{cases}\;\; . 
\end{equation}

From the incidence matrix, we can build even and odd graph Laplacians.  The even graph Laplacian (which is known as the combinatorial graph Laplacian)
\begin{equation} \label{eq: evenL}
    \Delta_+(\Gamma)= II^T
\end{equation} acts on vertex states from $\mathbb{C}^V\to \mathbb{C}^V.$ The following properties are well known for the even graph Laplacian \cite{Mohar_1991}.
\begin{theorem}\label{Th:Properties Delta+}
Let $\Delta_+$ be the even Laplacian of a finite simple graph.
\begin{enumerate}
    \item $\Delta_+$ is a real symmetric (Hermitian) matrix.
    \item The even graph Laplacian has the quadratic form \begin{equation} \label{eq:+_qform}
   \ q_{\Delta_+}=\sum_{v_i,v_j| v_iv_j\in E(\Gamma)} |v_i-v_j|^2. \end{equation}
    \item $\Delta_+$ is positive semidefinite. 
    \item The eigenvalues of $\Delta_+$ are non negative real numbers.
    \item $\Delta_+$ is independent of the orientation of the edges of $\Gamma$. 
\end{enumerate}
\end{theorem}

The even graph Laplacian $\Delta_+$ is a discrete (finite difference) approximation of the \textit{negative} continuous Laplacian.
We similarly define the odd graph Laplacian $\Delta_-$ to act on the edge states of a graph $\Gamma$ from  $\mathbb{C}^E\to \mathbb{C}^E$:
\begin{equation} \label{eq: OddL}
     \Delta_-(\Gamma)= I^TI.
\end{equation}
Unlike the even graph Laplacian, the odd Laplacian is orientation \textit{dependent}. That is, if we swap the direction of one of the edges in $\Gamma$, the odd Laplacian will change to reflect this. 

We can then define graph Schr\"odinger equations analogous to the continuum Schr\"odinger equation for a free particle, which describe how the (vertex and edge) graph states evolve. Let $\Gamma$  be a graph of order $m$ and  $|\psi\ket_v: V(\G) \rightarrow (\mathbb{C})$ a vertex state function that assigns a complex amplitude to each vertex. The even graph free Schr\"odinger equation is \begin{equation} \label{eq:Schrod+_G}
    \frac{\partial}{\partial t}|\psi\ket_v=i{\hbar}\Delta_{+}|\psi\ket_v,
\end{equation}
with solutions \begin{equation}
    |\psi(t)\ket_v=e^{i{\hbar}\Delta_+t}|\psi(0)\ket_v.
\end{equation}
For a graph $\G$ of size $n$ and and edge states $|\psi\ket_e: E(\G)\to(\mathbb{C})$, the odd graph free Schr\"odinger equation is 

\begin{equation} \label{eq:Schrod-_G}
    \frac{\partial}{\partial t}|\psi\ket_e=i{\hbar}\Delta_{-}|\psi\ket_e,
\end{equation} with solutions \begin{equation}
     |\psi(t)\ket_e=e^{i{\hbar}\Delta_-t}|\psi(0)\ket_e.
\end{equation}

\textit{Remark:} The apparent difference in the sign of the exponent between these and the continuous equations is corrected implicitly in the discrete Laplacians. Here we have also set $2m=1$, for simplicity.

\subsection{Steady States}\label{s:SteadSF_G}
As in the continuum , we say a graph quantum state is steady\footnote{A quantum state is called \textit{steady} if the probability distributions of all observables are independent of time. That is, when $
    \frac{\partial}{\partial t}|\psi\ket=0 \;\; \Leftrightarrow   e^{-i{\hbar}\hat Ht}|\psi\ket =  |\psi\ket $ holds in any basis. Then, the steady states of a system occur exactly when $E=0$ and are the vacuum energy eigenstates. There is an alternative way to define steady states (also called ``stationary'' states), with which the reader may be more familiar. In the ``projective picture,'' states related by a time-varying global phase are said to be equivalent, since phase is not observable. In this picture, all energy eigenstates may be considered ``steady''. However, we have chosen to define steady states to coincide only with vacuum states.} if 
\begin{equation}
    e^{i{\hbar}\Delta_+t}|\psi_0\ket_v=|\psi_0\ket_v \;\; \text{or } \;\;   e^{i{\hbar}\Delta_-t}|\psi_0\ket_e=|\psi_0\ket_e.
\end{equation}
To understand the topological conditions for achieving these steady states, we use a few key theorems.
\begin{theorem} \label{Thm:inKer}
    Let A be a $n\times n-$ matrix, $|\psi_0\ket \in \mathbb{C}^n$ an n-dimensional vertex or edge quantum state, $k$ a non-zero complex constant, and $t$ a complex parameter. Then, $e^{kAt}|\psi_0\ket=|\psi_0\ket$ for all values of $t$ if and only if $|\psi_0\ket \in \text{Ker}(A)$ . 
\end{theorem}
For proof, see \cite{Casiday_Contreras_Meyer_Mi_Spingarn_2024}. Thus, to characterize the steady states of the Laplacian (even or odd), we need to find their kernels. We'll begin with the even Laplacian.

\begin{theorem}\cite{Casiday_Contreras_Meyer_Mi_Spingarn_2024}
    \label{Thm:Ker+} A vertex quantum state is in the kernel of the even Laplacian if and only if it is constant on each connected component.
\end{theorem}

Finally, combining Theorems \ref{Thm:inKer} and \ref{Thm:Ker+}, we conclude the following: 
 \begin{theorem} \cite{Casiday_Contreras_Meyer_Mi_Spingarn_2024}
     A vertex state $|\psi\ket_v$ is steady for the even Laplacian if and only if it is constant on each of its connected components. 
 \end{theorem}

\begin{corollary}
    \label{Thm: DimKer+}  $\dim[\ker(\Delta_+(\G))]$ is the number of connected components of  $\G$ \cite{Contreras_Xu_2019}. 
\end{corollary}

\begin{theorem}
\label{Thm:Ker-}An edge quantum state is in the kernel of the odd Laplacian if and only if it is a linear combination of independent edge cycle states.
\end{theorem}
\begin{proof}

Let $\{c_1, c_2, \dots c_n\}$ be the independent cycles of a graph $\Gamma$ of size $m$. For each $c_i$, let $|c_i\ket$ denote the independent cycle edge state with value $1$ on clockwise oriented edges in $c_i$, $-1$ on counter-clockwise oriented edges, and $0$ on edges not in $c_i$. Call $c=\{|c_i\ket\}$ the set of linearly independent cycle edge states. It is sufficient to show that these states span the kernel of the odd Laplacian. 

Consider an edge state $|e_k\ket$ with value $1$ on edge $e_k$ and $0$ on all other edges. Let $v_{k_0}$ and $v_{k_f}$ be the end vertices of $e_k$, with $e_k$ originating from $v_{k_0}$. The action of the incidence matrix on $e_k$ is $$I|e_k\ket=|v_{k_f}\ket-|v_{k_0}\ket,$$
where $|v_{k_f}\ket$ and $|v_{k_0}\ket$ are the vertex states with value $1$ on vertices $v_{k_f}$ and $v_{k_0}$, and $0$
 elsewhere. Then, consider the action of $I$ on an independent cycle state $$|c_i\ket=|e_{i1}\ket+|e_{i2}\ket+\dots+|e_{ij}\ket,$$ where each $|e_{ik}\ket$ has value 1 if oriented clockwise and $-1$ if oriented counter-clockwise.  The negative sign on each counter-clockwise oriented edge can be thought of as flipping the orientation of that edge to clockwise, such that the final vertex of each edge $e_{ik}$ is the initial vertex of edge $e_{i(k+1)}$, with $v_{i(j_f)}=v_{i1}.$ 
 Then, by linearity, \begin{equation*}\begin{split}
         I|c_i\ket &=I|e_{i1}\ket+I|e_{i2}\ket+\dots+I|e_{ij}\ket \\&= (|v_{i1_f}\ket-|v_{i1_0})+ (|v_{i2_f}\ket-|v_{i2_0})+\dots+(|v_{ij_f}\ket-|v_{ij_0})\ket\\&=(|v_{i1_f}\ket-|v_{i1_0})+ (|v_{i2_f}\ket-|v_{i1_f})+\dots+(|v_{i1_0}\ket-|v_{i(j-1)_f})\ket=|0\ket.
 \end{split}
 \end{equation*}
(Notice: this is just Kirchhoff's first law: the ingoing and outgoing edges of a cycle must balance.) Thus, every vector representing a linear combination of independent edge cycle states is in the kernel of the incidence matrix and $\text{ker}(I)=\text{span}(c)$. Then, since $\text{ker}(\Delta_-)=\text{ker}(I)$, $\text{ker}(\Delta_-) = \text{span}(c)$. 
\end{proof}
It is worth noting that the steady states are orientation dependent. That is, two orientations of the same graph may have different spaces of steady states. Finally, combining Theorems \ref{Thm:inKer} and \ref{Thm:Ker-} , we conclude the following:

\begin{theorem}
    An edge state $|\psi\ket_e$ is steady for the odd Laplacian if and only if  it is a linear combination of independent edge cycle states.
\end{theorem}
\begin{corollary}
\label{Thm: dimker-} $\dim[(\ker\Delta_{-}(\G)]$) is the number of independent cycles of $\G$ \cite{Contreras_Xu_2019}. 
\end{corollary}\begin{proof} The proof follows directly from that of Theorem \ref{Thm:Ker-}. Since $\text{ker}(\Delta_-(\G))=\text{ker}(I(\G)),$ and $\text{span}(c)=\text{ker}(\Delta_-(\G))$,  it follows that $\text{dim[ker$(\Delta_{-}(\G))$]}$ is the number of linearly independent cycle states. 
\end{proof}

\subsection{Excited States} \label{s:ExitSF_G}
Understanding the non-zero eigenvalues of the graph Laplacian can help us to understand the excited states of physical systems. Unlike the steady states, which are purely topological (they depend on the general ``shape'' of the graph), the conditions for higher energy eigenstates are more difficult to determine and can depend on the geometry of the graph (for instance, how many vertices are in each cycle). While the non-zero eigenvalues for the adjacency matrix have been more well studied, Mohar argues that the spectrum of the graph Laplacian is more intuitive and more widely applicable \cite{Mohar_1991} . Certainly this true in our case. We present some important properties about the spectra of even graph Laplacians. Let the eigenvalues of $\Delta_+$ be $\lambda_1\leq\lambda_2\leq\dots\leq\lambda_n.$ Then:
\begin{enumerate}
    \item The smallest eigenvalue of $\Delta_+(\Gamma)$ is $\lambda_1=0$ \cite{Mohar_1992}. These are our steady states. 
    \item The second smallest eigenvalue $\lambda_{2}>0$ if and only if $\G$ is connected.\footnote{Naively, we can understand that if the graph is disconnected, the zero-eigenvalue is degenerate, so the second smallest eigenvalue is $0$. The rigorous proof in both directions follows from the Matrix Tree Theorem or Kirchhoff's Theorem \cite{Kirchhoff_1847}, which says that for a connected graph $\G$ with $n$ vertices, and non-zero eigenvalues listed in order of decreasing value $\lambda_{1}, \lambda_2\dots\lambda_{n-1}$, the number of spanning trees of $\G$ is $t(\G)=n^{-1}\lambda_1\lambda_2\dots\lambda_{n-1}$). } This eigenvalue, called the Fiedler value, is equivalent to the algebraic connectivity of the graph, a measure of how densely connected the graph is (see \textcite{Fiedler_1973, zhang_2011}). The Courant-Fisher Principle bounds the Fiedler value from above \cite{Mohar_1991}: $$\lambda_{2}\leq \min[(\Delta_+ x,x)/(x,x)],$$
where $x=|\psi\ket_v$ is a non-constant, non-trivial vertex state and $(\cdot,\cdot)$ is the inner product on $\mathbb{C}^n.$ The Fiedler value is also closely connected with a discrete version of the \textit{Cheeger constant} (also known as the isoperimetric number), a measure of the most efficient way to partition a Riemannian manifold into two pieces. In loose terms, the Cheeger constant $h_\G$ of a graph $\G$quantifies the minimum number of cut edges needed to divide a graph in half. On a connected graph, the discrete Cheeger inequality relates the Cheeger constant to the Fiedler value of the (degree-normalized) even Laplacian: \begin{equation} \label{eq:Cheeger}
    2h_\G\geq\lambda_2\geq \frac{h_\G^2}{2}.
\end{equation}
For formal definitions, proofs of the graph Cheeger inequality, and a neat overview of their importance in optimal partition algorithms for graphs, see Chung \cite{Chung_2010,Chung_2009}.
    \item Merris provides an upper bound on the spectral radius, the largest eigenvalue $\lambda_n$, of the even graph Laplacian \cite{Merris_1998}. Let $d(v)$ denote the degree of vertex $v\in\Gamma$ and $m(v)$ denote the average of the degrees of the vertices adjacent to $v$. Then $\lambda_n$ is bounded above by \begin{equation}
        \lambda_{n} \leq \text{max}\{m(v)+ d(v): v\in V(\G)\}.
    \end{equation} For a thorough discussion of the spectral radius of the graph Laplacian, and the many ways it can be constrained for different kinds of graphs, see \cite{Li_Zhang_1998, zhang_2011}. 
\end{enumerate}

\begin{theorem}
   \label{thm:+-_same_spectrum} The non-zero eigenvalues of even and odd graph Laplacians coincide.
\end{theorem}
\begin{proof}
Since $I$ is real valued, it follows  matrices $II^T$ and $I^TI$ are positive semi-definite and their non-zero eigenvalues coincide. 

To see that that are positive semi-definite, consider $$v^T(II^T)v=(v^TI)(I^Tv)=(I^Tv)^T(I^Tv)=||I^Tv||^2\geq 0$$ and $$v^T(I^TI)v=(v^TI^T)(Iv)=(Iv)^T(Iv)=||Iv||^2\geq 0.$$
To see that their non-zero eigenvalues coincide, let $\lambda\neq0$ be an eigenvalue of $I^TI$ with corresponding nontrivial eigenvector $v$. Then, $II^Tv=\lambda v.$ \\Multiplying on the left by $I$, $$I(I^TIv)=I\lambda v \Rightarrow(II^T)Iv=\lambda (Iv).$$Then, $Iv$ is a nontrivial eigenvector of $(II^T)$ with eigenvalue $\lambda.$ Thus every eigenvalue $\lambda\neq0$ of $(I^TI)$ is also an eigenvalue of $(II^T).$ We could show by a similar argument that every eigenvalue $\mu\neq 0$ of $(II^T)$ is also an eigenvalue of $(I^TI)$ with corresponding eigenvector $I^Tv.$

\end{proof}
\begin{corollary}
    \label{cor: o_indep_spectra} The spectra of the even and odd Laplacians are independent of the graph's orientation. 
\end{corollary}

\subsection{Dirac's Equation on Graphs} \label{s: Dirac_G}
Since there is a natural discrete analog of the Laplace operator on graphs, it follows that we might construct a similar graph theoretic analog for the Dirac operator. 

To fully appreciate the Dirac equation on graphs, we use \textit{vertex-edge graph states}, which index both the vertex and edge amplitudes in one vector. Accordingly, we also introduce \textit{vertex-edge walk} \cite{DelVecchio_2012,Yu_2017}. A vertex-edge walk $\gamma_{ve}$ \label{walk} of length $1$ moves from a vertex to an edge or an edge to a vertex (it does not matter which). If the edge involved has orientation originating from the vertex involved then $\text{sgn}(\gamma_{ve})=-1$. If the edge enters the vertex, then $\text{sgn}(\gamma_{ve})=1$. The sign of a vertex-edge walk of length $k$ is the product of the sign of each step. 

From this definition, we can describe two subtypes of walks, which we will call \textit{vertex superwalks} and \textit{edge superwalks}. A  vertex superwalk $\gamma_v$ of length $1$ starting on $v_i \in V(\Gamma)$ can go through an incident edge to an adjacent vertex, the typical characterization of walks. In this case, we say $\text{sgn}(\gamma_v)=-1$.  Or, it can go through an incident edge and return to vertex $v_i$, in which case  $\text{sgn}(\gamma_v)=1$. We call this consecutive repetition of a vertex in a superwalk a ``hesitation.''  An edge superwalk $\gamma_e$ of length $1$ starting at $e_i$ can go through an end vertex $v$ to an adjacent edge $e_j$. If both $e_i$ and $e_j$ are oriented into or out of $v$, $\text{sgn}(\gamma_e)=1$ and if one of $e_i$, $e_j$ is oriented  into and the other out of $e_j$,  $\text{sgn}(\gamma_e)=-1$. Or,  $\gamma_e$ can go through $v$ back to edge $e_j$, in which case $\text{sgn}(\gamma_e)=1$. Vertex superwalks and edge superwalks are vertex-edge walks of even length. 

The discrete Dirac operators are formal square roots of the discrete Laplace operators.\footnote{While in the continuum the Dirac operator is a sort of (multi-dimensional) ``square root'' of the d'Alembertian, without a metric or embedding space on our graph to distinguish space-like and timelike dimensions, there is no distinction between the d'Alembertian and the four dimensional Laplacian.} Since the graph Laplacians are real symmetric matrices, they are diagonalizable. We will write the diagonalizations of the even and odd Laplacians as  $P_+D_{\Delta_+}P_+^{-1}$ and $P_-D_{\Delta_-}P_-^{-1},$ respectively, where $\sqrt{D_{\Delta_\pm}}$ is the square root of each of their diagonal entries. Then, we define even and odd Dirac operators, acting on vertex states and edge states, respectively as: \label{edir} \label{odir} \begin{equation}
    \slashed{D}_+=P_+\sqrt{D_+}P_+^{-1},
\end{equation} and \begin{equation}
    \slashed{D}_-=P_-\sqrt{D_{\Delta_-}}P_-^{-1}.
\end{equation}

These can be seen as square roots of the Laplacians: \begin{equation}
    (\slashed{D}_\pm)^2=P_\pm\sqrt{D_{\Delta_\pm}} P_\pm^{-1}P_\pm\sqrt{D_{\Delta_\pm}} P_\pm^{-1}=\Delta_\pm.
\end{equation}

Following the form of the Laplace operators, we can define the \textit{Incidence Dirac Operator} to act on vertex-edge states as \begin{equation}\label{eq:D_G}
    \slashed{D}_I=\begin{pmatrix}
        0 & I\\I^T & 0
    \end{pmatrix}.
\end{equation} This is a formal square root of the direct sum of the even and odd Laplace operators: \begin{equation} \label{eq:D^2_G}
      (\slashed{D}_I)^2=\begin{pmatrix}
        \Delta_+ & 0\\0 & \Delta_-
        \end{pmatrix}.
\end{equation}
The quadratic form of the Dirac Incidence operator is \begin{equation}
    q_{\slashed{D}_I}\begin{pmatrix}
        v\\ e
    \end{pmatrix}= 2\sum_{e_{i,j}\in E(\G)} e_{i,j} (v_j - v_i).
\end{equation} 
Casiday et al. show that the elements of the powers of the graph Dirac operator are related to vertex-edge walks, in a similar manner to which powers of the adjacency matrix of graphs are related to standard walks on graphs:
\textcite{Casiday_Contreras_Meyer_Mi_Spingarn_2024}. Specifically, \begin{equation} \label{eq:D_k_walk}
    [\slashed{D}_I]^k_{ij}= \sum_{\gamma, i\to j,k} \text{sgn}(\gamma_{ve}).
\end{equation}
Finally, we use our incidence Dirac operator to define the Dirac equation on graphs as \begin{equation}
    \label{eq: GraphDirac} (i\slashed{D}_I-m)|\psi_{ve}\ket=0,
\end{equation} where $|\psi_{ve}\ket$ is a vertex-edge state. 

As with the Schr\"odinger equation, we want to determine the steady states of the Dirac equation. By Theorem \ref{Thm:inKer}, this amounts to finding the kernel of the Dirac operator.

\begin{theorem} \cite{Casiday_Contreras_Meyer_Mi_Spingarn_2024} The following holds for the graph Dirac and Laplace operators:
  \newline  
  (1) $\ker(\slashed{D}_\pm)=\ker(\Delta_\pm)$ \\(2) $\ker(\slashed{D}_I)=\ker(\slashed{D}_+\oplus\slashed{D}_-).$
\end{theorem}

\subsection{Path Integrals on Graphs} \label{s:PI_G}
In the continuum, computing the path integral is very challenging most of the time. We can't even always be certain that it converges. A discrete analogy and convenient simplification of the path integral is naturally defined on graphs, where the integral becomes an absolutely convergent sum. Mnev \cite{Mnev_2007, Mnev_2016} outlines how to set up these discrete path integrals, which we will call \textit{walk sums}.\footnote{This choice of term is important: in our graph theoretic terminology, a path from one vertex to another is a sequence of \textit{distinct} intermediate vertices. On the other hand, a \textit{walk} allows for vertex repetition. Thus it is the graph theoretic ``walk'' rather than the graph theoretic ``path'' that is analogous to the quantum mechanical path.}

Let $\Gamma$ be a finite graph and let $\{|v\br\}_{v \in V}$ be a set of  vertex basis states of  $\Gamma.$Then, consider the matrix element of the matrix exponential $$\bl v_i |e^{tA}|v_j\br = \sum_{n=0}^\infty \frac{1}{n!}t^nA^n,$$where, $v_i$ and $v_j$ are are two vertices of $\Gamma$, $t$ is a real parameter, and $A$ is the adjacency matrix (Equation \ref{eq:adjacency mtx}). By Theorem, \ref{Th:number_of_walks}, $[A^n]_{ij}$ is the number of possible walks from $v_j$ to $v_i$ of length $n$. Then, $\sum_{n=0}^\infty A^n$, is the total number of possible walks from $v_j$ to $v_i$ on $\G$. Thus, we may rewrite each matrix element as a sum over all possible walks $\gamma$ from $v_j$ to $v_i$: \begin{equation}
    \bl v_i|e^{tA}|v_j\br = \sum _{\text{walks} \; \gamma_{v_j v_i}}\frac{t^{|\gamma|}}{|\gamma|!},
\end{equation}
where $|\gamma|$ is the length of the walk $\gamma$. 

We first examine the special case of this sum on regular graphs. 
For a $k$-regular graph, the degree matrix is  $D=k\mathbb{1}$, where $\mathbb{1}$ is the identity matrix. Thus, the even Laplacian will be $\K=A-k\mathbb{1}$ and we can write $$e^{\K t}= e^{(A-k\mathbb{1})t}=e^{At}e^{-kt}.$$
Then, the matrix element \begin{equation} \label{eq: graph_propagator}
    \bl v_i|e^{t\K}|v_j\br = \sum_{k=0}^\infty\frac{t^{k}}{k!} e^{-kt}W_{k, ij}=\sum _{\text{walks} \; \gamma_{v_j v_i}}\frac{t^{|\gamma|}}{|\gamma|!} e^{-kt},
\end{equation}
where $W_{k, ij}$ denotes the number of walks of length $k$ from $v_i$ to $v_j$. Compare this to the continuous path integral. The left-hand side of this equation looks like the like the propagator.  The absolutely convergent sum over walks on the right-hand side looks like a discrete Euclidean path integral, with measure $\mathcal{D}[q(t)]=\frac{t^{|\gamma|}}{|\gamma|!}$ and Euclidean action $S_E[q(t)]= kt.$ Intuitively this makes some sense. In the measure, the factorial term dominates, so longer walks have smaller contributions. This is similar to the classical paths being the geodesics in the continuum. However, the number of incident edges on each vertex $k$ is only the graph action on regular graphs. Let's examine case of non-regular graphs. To do so we must consider vertex superwalks. \textcite{Yu_2017} shows \footnote{The proof is done by induction on $k$.} that for $\gamma_+$ a vertex superwalk of length $k$ from vertex $v_i$ to $v_j,$\begin{equation}
    [\K^k]_{ij}=\sum_{\gamma,i\to j,k} \text{sgn}(\gamma_+).
\end{equation}
For vertex superwalks, $\text{sgn}(\gamma_+)$ is easy enough to visualize because it is independent of the orientation of the edges of the graph and depends only on the length and number of hesitations of the walk. 

\textcite{DelVecchio_2012} shows that using the vertex superwalk definition, a general version of the walk sum formula is: 
\begin{equation} \label{eq:gen_g_prop}
    \bra v_i|e^{t\Delta_+}|v_j\ket= \sum _{k=0}^\infty\frac{t^{k}}{k!} (-1)^k\sum_{\gamma\in W_{k,ij}} deg(\gamma),
\end{equation}
where $$deg(\gamma)=(-1)^{\#(n|v_n\neq v_{n+1}\text{\;for\;}1\leq n\leq k, \; v_n, v_{n+1}\in \gamma)}$$
measures non-hesitant steps in the walk. This holds on any graph and reduces to Equation \ref{eq: graph_propagator} in the case of regular graphs.

A similar analysis might be applied to investigate edge superwalks, even though the odd Laplacian does not have nearly so nice a relationship to the adjacency matrix as the even Laplacian does. Similar to the even Laplacian case, Yu proves in \cite{Yu_2017} that for $\gamma_-$ an edge superwalk,
\begin{equation}
    [\Delta_-]_{ij}^k=\sum_{\gamma,i\to j,k} \text{sgn}(\gamma_-).
\end{equation}

We can also recover the graph partition function, by taking the trace of the graph Euclidean propagator (Equations \ref{eq: graph_propagator} and \ref{eq:gen_g_prop} ) \cite{Mnev_2007}. This is equivalent to considering the sum over closed walks.  On a $k$-regular graph
\begin{equation} \text{tr}(e^{t\Delta_+})= \bl v_i|e^{t\K}|v_i\br =     \sum_{\text{closed} \gamma} \frac{t^{|\gamma|}}{|\gamma|!} e^{-kt} . 
\end{equation}
 In general, this becomes, 
\begin{equation}
    \bra v_i|e^{t\Delta_+}|v_i\ket= \sum _{k=0}^\infty\frac{t^{|\gamma|}}{|\gamma|!} (-1)^k\sum_{\gamma\in W_{k,ii}} deg(\gamma).
\end{equation}

\section{Supersymmetry in the Continuum} \label{SUSY}
\subsection{Superalgebras}
We move now to the fundamentals of (continuum) supersymmetry which we seek to translate into the graph quantum mechanical picture outlined above. The mathematical structures of supersymmetric theories are called \emph{superalgebras}. We primarily adopt the superalgebra conventions of Witten \cite{Witten_1982}, with influences and adaptations from Tong \cite{Tong_2022} and Wegner \cite{Wegner_2016}. 

In a supersymmetric theory, fermions and bosons can be grouped into \textit{supermultiplets}. The constituent bosons and fermions of each supermultiplet are called \textit{superpartners}. Mathematically, this means that a set of Hermitian symmetry operators $Q_1, Q_2, \dots, Q_N$, called \textit{supercharges}, exist, where each $Q_i$  maps bosonic states to fermionic states and vice versa.

Consider a Hilbert space of states $\mathcal{H}$. We can express $\mathcal{H}$ in terms of its constituent Hilbert spaces of bosonic and fermionic states, $\mathcal{H}^+$ and $\mathcal{H}^-$, respectively, where $\mathcal{H}$ is the direct sum $\mathcal{H}=\mathcal{H}^+\oplus \mathcal{H}^-$.  Then, we may express a state $|\psi\ket $ in $\mathcal{H}$ as \begin{equation}
      |\psi\ket = \begin{pmatrix}
      B\\ \hline F
    \end{pmatrix},
\end{equation} where $B\in \mathcal{H^+}$ and $F\in\mathcal{H}^-$. Since a supercharge maps between bosonic and fermionic states, it is of the form
\begin{equation}
  Q_i =  \left (\begin{array}{c|c}
    \  0 &  M^\dagger\\ \hline
     \ M & 0\\ 
 \end{array} \right ),
\end{equation}  such that  \begin{equation}
    Q_i  \begin{pmatrix}
        B\\ \hline F
    \end{pmatrix} = \left (\begin{array}{c|c}
    \  0 &  M^\dagger\\ \hline
     \ M & 0\\ 
 \end{array} \right ) \begin{pmatrix}
        B\\ \hline F
    \end{pmatrix} = \begin{pmatrix}
        M^\dagger F\\ \hline MB
    \end{pmatrix},
\end{equation}
with $ M^\dagger F  \in \mathcal{H}^+$ and  $M B \in \mathcal{H}^-.$ That is, $Q_i$ can be decomposed into a sum of adjoint operators $M$ and $M^\dagger$ (with only mild notational abuse), where $M$ maps bosonic states to fermionic states and $M^\dagger$ maps fermionic states to bosonic states. Notice that $M^\dagger$ must be the adjoint of $M$ to satisfy the Hermitian property of $Q_i$, and $M^2={M^\dagger}^2=0$.

Three properties describe supercharges.

\begin{enumerate}
    \item Supercharges commute with the Hamiltonian: \begin{equation}
    [ Q_i, H]= Q_i H- H Q_i=0.
\end{equation} Thus, we can understand supercharges as conserved quantities, indicative of an underlying symmetry, the ``supersymmetry''. An important consequence of this commutation is that the supercharges do not change the mass of the state on which they act, and the superpartners should have the same mass. 
    \item Supercharges anti-commute with the fermion parity operator. Recall that fermions and bosons differ in their possible spins. Bosons have integer spins, while fermions have half integer spins. The \textit{fermion parity operator }$$(-1)^F=e^{2\pi i J_z}$$ distinguishes the bosonic states in $\mathcal{H}^+$ from the fermionic states in $\mathcal{H}^-$, where $F$, the \textit{fermion number}, is $0$ if the state is bosonic and $1$ if the state is fermionic and $J_z$ is the angular momentum. Under the action of this operator, a bosonic state $|b\ket$ will remain unchanged, while a fermionic state $|f\ket$ acquires a negative sign:\footnote{``Bosonic states'' and ``fermionic states'' need not consist of a single fermion or a single boson, but instead can have a net integer or net half-integer spins, respectively, as distinguished by the fermion parity operator.} \begin{equation} \begin{split}
    \left(-1\right)^F|b\ket &=|b\ket \; \; \;\;\ \text{for all} \ |b\ket\in\mathcal{H}_+ \\ \left(-1\right)^F |f\ket &=-|f\ket \; \;  \text{for all} \ |f\ket\in\mathcal{H}_- 
\end{split}
\end{equation}
\begin{equation}
    \{Q_i, (-1)^F\} = (-1)^F Q_i +  Q_i(-1)^F=0.
\end{equation}
\item For a normalized basis of $N$ supercharges, \begin{equation}
       Q_i^2=H  
\end{equation} \begin{equation} \label{eq: Q^2=H}
    \{Q_i,Q_j\}=Q_iQ_j+Q_jQ_i=0 \;\; \text{if } \;\ i\neq j .
\end{equation}Note that to ensure Lorentz invariance of this theory for nonzero momentum, a momentum term would appear in the superalgebra on the right-hand side of Equation \ref{eq: Q^2=H}. However, when we consider spontaneous supersymmetry breaking, we will search for vacuum states of zero-energy that are annihilated by the supercharges. We will see that supersymmetry implies that $E\geq 0$  and each zero-energy state also has zero momentum. Thus, we will find it entirely sufficient to restrict our study to the subspace of our Hilbert space containing only the zero-momentum states, where the algebra reduces to the relations above. 
\end{enumerate}

We consider the simplest case of $(1+1)$-dimensions, with two supercharges $Q_1$ and $Q_2$, defined such that \begin{equation} \label{eq: Q1_Q2_M}
    Q_1=M+M^\dagger \;\; \text{and} \;\  Q_2=i(M-M^\dagger).
\end{equation}

From (3) we see that the supersymmetric Hamiltonian is related to the supercharges by \begin{equation} \label{eq: SUSY_H}
  H_S=  \frac{1}{2} (Q_1^2+Q_2^2) =( M^\dagger M+M M^\dagger).
\end{equation}
 Accordingly, we can break this Hamiltonian into its bosonic and fermionic sectors. Notice, $M^\dagger M$ acts on bosonic states from $\mathcal{H}^+\rightarrow \mathcal{H}^+$ while $MM^\dagger$ acts on fermionic states from $\mathcal{H}^-\rightarrow \mathcal{H}^-$.  Then, \begin{equation} \label{eq: SUSY_H_2}
     H_S=H^++H^- = \left (\begin{array}{c|c}
    \  H^+ &  0\\ \hline
     \ 0 & H^-  \\ 
 \end{array} \right )=\left (\begin{array}{c|c}
    \  M^\dagger M &  0\\ \hline
     \ 0 & M M^\dagger \\ 
 \end{array} \right ),
 \end{equation}
 where $H^+=M^\dagger M$ and $H^-=M M^\dagger $. 

The spectrum of the supersymmetric Hamiltonian has two unique and important properties. 

\begin{theorem}
    \label{thm: D_SUSY_H} Non-zero eigenvalues of SUSY Hamiltonians have even degeneracies. More specifically, superpartners of $Q_i$ have the same energy.
\end{theorem}
\begin{proof}
Let $|b\ket$ and $|f\ket$ be superpartners such that $Q_i|b\ket=|f\ket$ and  $Q_i|f\ket=|b\ket$. Let $H|b\ket=E^+|b\ket$ and $H|f\ket=E^-|f\ket$. Then, \begin{equation*}
    \begin{split}
        H|b\ket=HQ_i|f\ket=Q_iH|f\ket=Q_iE^-|f\ket=E^-Q_i|f\ket=E^-|f\ket.
    \end{split}
\end{equation*}Then, $$E^+=E^-.$$  (In reality,  to obtain the proper normalization, since $Q_i^2=H$, we want $Q_i|b\ket=\sqrt{E}|f\ket$ and $Q_i|f\ket=\sqrt{E}|b\ket$).

This pairing doesn't hold, however, in the case of $E=0$. Since $Q_i^2=H$, if $H|b\ket=0$, then $Q_i|b\ket=0\neq|f\ket$, since the state is simply annihilated by the supercharge. (Likewise for $H|f\ket=0$). Notice that zero-energy states are in the kernel of the supercharges, and are therefore in their own supermultiplets.
\end{proof}
\begin{theorem}
    \label{thm: PD_SUSY_H} The Hamiltonian is positive semi-definite, and therefore the energy of a SUSY Hamiltonian is non-negative. 
\end{theorem}
\begin{proof}
Consider the vacuum state $|\Omega\ket.$ If all of the supercharges annihilate the vacuum state, $Q_i^2|\Omega\ket=0$ for all $i$, then $\bra\Omega|H|\Omega\ket=\frac{1}{2}\bra\Omega|\sum_i Q_i^2|\Omega\ket= 0$. Otherwise, Suppose $Q_i|\Omega\ket\neq 0$ for some $j$. Then, $$\bra\Omega|H|\Omega\ket=\frac{1}{2}\bra\Omega|\sum_j Q_j^2|\Omega\ket= \frac{1}{2}\sum_j\bra\Omega|Q_iQ_i|\Omega\ket=|Q_i|\Omega\ket|^2>0.$$

    This is in contrast to non-supersymmetric quantum mechanics, where a constant can simply be added to the Hamiltonian to shift the zero-energy state. In a supersymmetric picture, the energy of the vacuum is fixed by Equation \ref{eq: SUSY_H}. 
\end{proof}

\subsection{Spontaneous Supersymmetry Breaking} \label{s: SSB_c}
If a symmetry is \textit{spontaneously} broken, the Lagrangian of the system respects the symmetry, but by varying the free parameters of the theory, one can obtain vacuum state solutions which shifted from zero such that they are not invariant under the symmetry operator. Thus, the symmetry becomes hidden.\footnote{This is in contrast to explicit supersymmetry breaking, where there are terms in the Lagrangian of the theory that don't respect the symmetry, but so long as the terms are small, the symmetry is approximate.} In the case of supersymmetry, this spontaneous breaking  allows for the appearance of a non-supersymmetric physical world, with superpartner pairs that differ in mass from each other, despite an underlying supersymmetric structure of the theory.  Since we don't see see boson-fermion pairs in nature, in order for supersymmetry to be a real symmetry of nature, it must be spontaneously broken. 

This occurs when the supercharges do not annihilate the vacuum state, ie. when \begin{equation}
    Q_i|\Omega\ket\neq 0 \Leftrightarrow  Q_i^2|\Omega\ket=H|\Omega\ket \neq 0,
\end{equation} for $|\Omega\ket$ the lowest energy state of the system. 
For a supersymmetric Hamiltonian, Theorem \ref{thm: PD_SUSY_H} tells us that the energy eigenvalues cannot be negative. Then, we state the following condition:
\begin{definition} Supersymmetry is \textit{spontaneously broken} if and only if the vacuum state energy is positive. That is, when supersymmetry is spontaneously broken, there exist no zero-energy bosonic states \textit{and} no zero-energy fermionic states. \end{definition}

\textit{When is the vacuum energy of a supersymmetric system greater than zero?} To begin to answer this question, we can look for criteria for theories for which the vacuum energy in any finite volume is exactly zero, following the topological approach of Witten \cite{Witten_1983, Witten_1982_2}. Such theories, which do not spontaneously break supersymmetry, can be eliminated as descriptions of supersymmetry in nature. Since the large volume limit of  $E=0$ is $E=0$,  knowing that supersymmetry is unbroken in any finite volume ensures that supersymmetry is unbroken in the infinite volume limit. This also justifies our restricted focus to finite graphs in the paper. 

As we vary the free parameters of our supersymmetric theory, such as the volume or coupling constants, the energy states shift. In particular, the vacuum energy states can be shifted above $E=0$, but they must move in supermultiplets, with a bosonic state and fermionic state shifting together, even though they were not related to one another by the supercharges when $E=0$.  Then, even when supersymmetry is spontaneously broken, the even degeneracy of $E>0$ energy states remains, although the masses of boson-fermion pairs need to longer be equal.\footnote{ The full explanation of this relies on the theoretical appearance of a Goldstino (a massless fermion) when supersymmetry is broken. See \textcite{Tong_2022_2}. }  

Figure \ref{fig:SUSY_shifting} shows three supersymmetric configurations of energy states, divided into bosonic states (blue $\text{x}$'s) and fermionic states (green dots). Notice that in each configuration, the $E>0$ states are degenerate. In the leftmost configuration, there are two bosonic zero-energy states and one fermionic zero-energy state, so supersymmetry is unbroken. With one fewer fermionic zero energy state than zero-energy bosonic state, it is impossible for all of the-zero energy states to be shifted to positive energy in boson-fermion pairs, and thus this configuration is protected against spontaneous supersymmetry breaking. The middle configuration, with an equal number of zero-energy bosonic and fermionic states, also does not spontaneously break supersymmetry, but could be shifted into the rightmost configuration, which does. 
\begin{figure}
    \centering
    \includegraphics[width=0.5\linewidth]{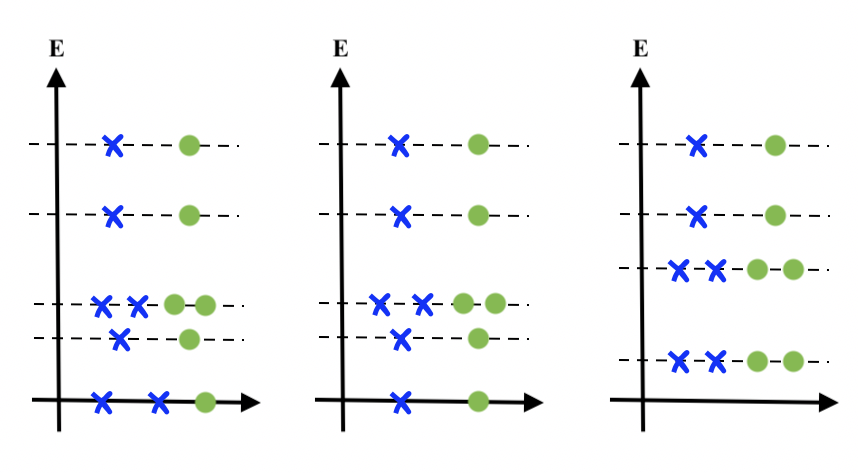}
    \caption{\small{Supersymmetric energy states, with bosons represented by blue $\text{x}$'s and fermions represented by green dots. Only the rightmost configuration spontaneously breaks supersymmetry, with a vacuum energy above $E=0$. The parameters of the leftmost system cannot be varied to break supersymmetry, due to the discrepancy in the number of zero-energy bosonic and fermionic states, while those of the middle configuration can.}}
    \label{fig:SUSY_shifting}
\end{figure}

Since bosonic and fermionic states shift in pairs, even as the parameters are varied to change the energies, the difference in the number of zero energy bosonic states minus the number of zero energy fermionic states is preserved. We call this invariant number the \textit{Witten index}, \label{win} denoted \begin{equation} \label{eq: W_Index_0}
    W=n_B^{E=0}-n_F^{E=0}, 
\end{equation} where $n_B^{E=0}$ and $n_F^{E=0}$ denote the number of zero-energy bosonic and fermionic states, respectively  \cite{Witten_1982_2}. 

A Witten index of $0$ is a \textit{necessary condition} for supersymmetry breaking. That is, if $W\neq0$, supersymmetry is unbroken. The argument for this is simple: if $W\neq0$, either $n_B^{E=0} > 0$ or $n_F^{E=0} > 0 $, or both, which implies the existence of at least one zero-energy state. However, knowing $W=0$, is not \textit{sufficient} to determine that supersymmetry is broken. Rather, there are two scenarios which give $W=0$. 

\begin{enumerate}[label=(\Alph*)]
    \item $n_B^{E=0}=n_F^{E=0}=0$ in which supersymmetry is spontaneously broken.
    \item $n_B^{E=0}=n_F^{E=0}\neq0$ in which $W=0$ but supersymmetry is unbroken.
\end{enumerate}

In the examples in Figure \ref{fig:SUSY_shifting}, the rightmost configuration satisfies $W=0$ (Scenario A), with SUSY spontaneously broken, while the middle satisfies $W=0$, without breaking supersymmetry (Scenario B). 

The Witten index is equivalent to \begin{equation} \label{eq: W_Index_Tr}
    W=\text{Tr}[(-1)^Fe^{(-\beta H)}].
\end{equation}
The trace term is a sum over a $+1$ for each bosonic state and a $-1$ for each fermionic state.
Since each non-zero energy state corresponds to a supermuliplet containing a boson and a fermion, the contributions to the trace from these states cancel out, leaving only the contributions from the states annihilated by the supercharges. For an infinite sum over states in our Hilbert space, the trace term is divergent. So, we regulate the sum with the statistical mechanical partition function $e^{-\beta H}$. Note that the Witten index for a supersymmetric system is independent of the parameter $\beta$, since the contributions to $W$ come only from terms where $E=0$, which sends the partition function to $1$.  

\subsection{Supersymmetric Path Integrals} \nocite{Tong_2022} \label{s:susy_PI}
We give here a brief discussion of supersymmetric Lagrangians and path integrals in the case of quantum mechanics, and refer the reader to \textcite{Tong_2022, Witten_1983, Witten_1981, Wegner_2016} for more detail. 

The spin-statistics theorem tells us that the wavefunctions of particles with integer spins (bosons) commute while the wavefunctions of particles with half-integer spins (fermions) anti-commute. To account for this, we introduce the anti-commuting \textit{Grassmann variable}, $\psi$, for which \begin{equation}
     \{\psi,\psi^\dagger\}= 1 \;\;\; \text{and} \;\;\; \{\psi,\psi\}= \{\psi^\dagger,\psi^\dagger\} = 0. \end{equation} 
First, consider a bosonic system. In this case, the fermion parity operator $(-1)^F=1$ and the Witten index\footnote{Technically,  $\text{Tr}[(-1)^Fe^{-\beta H}]$ is only an ``index'' in a supersymmetric system. However, the quantity may be computed regardless and we will choose to call it the Witten index in both cases to avoid unnecessary complication.} (Equation \ref{eq: W_Index_Tr}) simplifies to $$W_B=Z_B=\text{Tr}[e^{-\beta H}],$$ the familiar thermal partition function$Z_B$ can be computed with a Euclidean path integral over closed paths, implementing periodic boundary conditions $\mathbf{q}(0)=\mathbf{q}(\beta).$

To preserve supersymmetry in path integral calculations, we must apply the same boundary conditions for the bosons and fermions.  Since the bosonic system is always periodic, the fermionic system must also be periodic.

Unlike in the bosonic case, for a fermionic system, $(-1)^F=-1$, so the Witten index differs from the thermal partition function by a minus sign. Implementing periodic boundary conditions  $\psi(0)=\psi(\beta)$ in the path integral approach here yields the Witten index \footnote{Computing the partition function $Z_F=\text{Tr}[e^{-\beta H}] = \int_{\psi(0)={-\psi(\beta)}}\mathcal{D}\psi^\dagger\mathcal{D}\psi e^{-S_E[\psi,\psi^\dagger]}$ requires implementing anti-periodic boundary conditions, $ \psi(0)=-\psi(\beta)$. Thus, in a SUSY theory it is simpler to compute the Witten index than the thermal partition function.}$$W_F=\text{Tr}[(-1)^Fe^{-\beta H}] = \int_{\psi(0)={\psi(\beta)}}\mathcal{D}\psi^\dagger\mathcal{D}\psi e^{-S_E[\psi,\psi^\dagger]}. $$

Consider the case of a particle moving in a potential on a line. The Hilbert space is \begin{equation}
    \mathcal{H}= L^2(\mathbb{R})\otimes\mathbb{C}^2,
\end{equation} where $L^2(\mathbb{R})$ are the set of position-basis states on the real line, and $\mathbb{C}^2$ are the spin degrees of freedom which distinguish fermions from bosons. The supercharges are \begin{equation} \label{eq:line_Qs}\begin{split} 
    Q_1&=\frac{1}{2}(\sigma_1p+\sigma_2h'(q))\\
    Q_2&=\frac{1}{2}(\sigma_2p-\sigma_1h'(q)),
\end{split}
\end{equation}
with $p = -id/dq$ the momentum, $\sigma_i$ the Pauli matrices, and $h(q)$ some function. Then, by Equation  \ref{eq: SUSY_H}, \begin{equation} \label{eq:SUSY_H_line}
    H=(p^2+h'(q)^2)\mathbb{1} - \sigma_3h'' (q).
\end{equation}
We will assume that $|h'(q)|\rightarrow \infty$ and $|q|\rightarrow \infty$ so that the spectrum of the Hamiltonian is discrete. 
The first term looks like the classical energy (setting $2m=1$ as before), with a potential term $$V(q)=\left(\frac{dh}{dq}\right)^2,$$ while the second term distinguishes the spin of the particle. 

The following is the action for the system in terms of both bosonic variables $q$ and fermionic Grassmann variables $\psi$ and $\psi^\dagger$: \begin{equation} \label{eq: SUSY_line_S}
     S[q,\psi,\psi^\dagger]=\int dt \;L =\int dt \; (\dot q^2 + i\psi^\dagger\dot\psi-h'^2+2h''\psi^\dagger\psi).
\end{equation}
Notice that using the Legendre transform $$H=p\dot q+\pi\dot\psi-L,$$ where $p(t)=2\dot q(t)$ is the conjugate momentum for the bosonic variables and $\pi(t) = i\psi^\dagger$  is the conjugate momentum of the fermionic Grassmann variables, we recover the Hamiltonian \begin{equation}
    H=(p^2+h'^2)-h''(\psi^\dagger\psi-\psi\psi^\dagger).
\end{equation} This is the same as the Hamiltonian in \ref{eq:SUSY_H_line}, where $$\psi\rightarrow\begin{pmatrix}
    0&1\\0&0
\end{pmatrix}\;\; \text{and}\;\; \psi^\dagger \rightarrow\begin{pmatrix}
    0&0\\1&0
\end{pmatrix},$$ such that $\psi^\dagger\psi-\psi\psi^\dagger=\sigma_3$. Thus, the action in Equation \ref{eq: SUSY_line_S} is justified. 

We can now find the Witten index. We apply a Wick rotation to compute the Euclidean action around closed paths \begin{equation}
     S_E[q,\psi,\psi^\dagger]=\oint dt \; \dot q^2 + \psi^\dagger\dot\psi+h'^2- 2h''\psi^\dagger\psi.
\end{equation}
Then, 
\begin{equation}
W=\text{Tr}[(-1)^Fe^{-\beta H}] =\int\mathcal{D}q\mathcal{D}\psi^\dagger\mathcal{D}\psi \;e^{-S_E[q,\psi,\psi^\dagger]}.
\end{equation}
The Witten index is independent of the magnitude of the potential $h$  \cite{Tong_2022}. That is, if we compute the Witten index rescaling $h\rightarrow \lambda h$ for some value $\lambda>0$, $dW/d\lambda = 0$. Then, we calculate $W$ in the infinite $\lambda$ limit. In this limit, the $(\lambda h)'^2$ term dominates the action, suppressing the contributions from all other terms, except in the case that $q(t)= a$ for which $h'(a)=0$. Thus, the infinite dimensional functional integral $W$ depends only on the contributions from the finite number of critical points of $h$. Then, we compute $W$ around these critical points. 

The behavior of functions in the close vicinity of critical points are well approximated by harmonic oscillators in the infinite-$\lambda$ limit: \begin{equation}
    \begin{split}
        h(q)&\approx h(a)+h'(a)(x-a)+\dots \\
        \Rightarrow h'(q)&\approx h''(a)(x-a) +\dots\\
        \Rightarrow V(q)&\approx h'' (a)^2(q-a)^2 +\dots \;\;\;,
    \end{split}
\end{equation}  with  $\omega= 2h''(a)$. The partition function (and therefore also the Witten index) for the bosonic harmonic oscillator is $$W_B= \int \mathcal{D}q\; e^{-S_E[q]}=  \frac{1}{2\text{sinh}(\beta|\omega|/2)}.$$ The absolute value reminds us that for the bosonic oscillator, $\omega>0$. (This is not true below in the fermionic case, where the anti-commutative nature of the Grassmann variables allow for negative $\omega$).  The Witten index of the fermionic oscillator is $$W_F=\int \mathcal{D}\psi^\dagger\mathcal{D}\psi\;e^{-S_E[\psi^\dagger, \psi]}=- 2\;\text{sinh}( \beta\omega/2).$$
Putting it all together, we find \begin{equation} \begin{split}
        W &= \int\mathcal{D}q\mathcal{D}\psi^\dagger\mathcal{D}\psi \;e^{-S_E[q,\psi,\psi^\dagger]} = \sum_a \frac{ - 2\; \text{sinh}(\beta h''(a))}{2\;\text{sinh}(\beta|h''(a)|)}  =\sum_a \frac{-\text{sinh}(h''(a))}{\text{sinh}|h''(a)|}.
\end{split}
\end{equation}
Since 
\begin{equation}
\frac{ - \; \text{sinh}(h''(a))}{\text{sinh}(|h''(a)|)}=\begin{cases}
            -1  & h''(a)>1 \\ 1  & h''(a) <0\\ \text{undefined} & h''(a)=0
        \end{cases}  \;\;\; ,
\end{equation}
\begin{equation}
    W=\text{Tr}[(-1)^Fe^{-\beta H}]=\sum_a \text{sign}(-h''(a)),
\end{equation}
an alternating sum of $-1$'s and $1$'s over the critical points of the potential, just as we asserted before.

\begin{example} \textbf{Particle on a Circle} \label{ex:Particle_on_circle}
As a simple illustration which we will demonstrate can be approximated on graphs in Section \ref{Results}, we show here the zero-energy vacuum state solutions for a particle moving on a circle $S^1$ of radius $R$. The supercharges (Equation \ref{eq:line_Qs}) and Hamiltonian (Equation \ref{eq: SUSY_line_S}) are the same as for a particle moving on a line. To satisfy the boundary conditions of the circle, we impose the restriction $h(q)=h(q+2\pi R)$. The zero-energy states must be annihilated by the supercharges, so must satisfy \begin{equation*}\begin{split}
   Q_1|\Omega\ket&=\left[\begin{pmatrix}
        0 & -i\frac{d}{dq}\\ -i\frac{d}{dq} & 0
    \end{pmatrix}+\begin{pmatrix}
        0&-ih'(q)\\ih'(q) &0
    \end{pmatrix}\right]\begin{pmatrix}
        B\\F
    \end{pmatrix}=0, \\
       Q_2|\Omega\ket&=\left[\begin{pmatrix}
        0 & -\frac{d}{dq}\\ \frac{d}{dq} & 0
    \end{pmatrix}+\begin{pmatrix}
        0&- h'(q)\\-h'(q) &0
    \end{pmatrix}\right]\begin{pmatrix}
        B\\F 
    \end{pmatrix}=0 . 
\end{split}
\end{equation*}
We find two linearly independent solutions to these equations, one bosonic and the other fermionic: \begin{equation*}
    b(q)= \begin{pmatrix}
        e^h\\0
    \end{pmatrix}\;\;\; \text{and} \;\;\; f(q)= \begin{pmatrix}
        0\\e^{-h}
    \end{pmatrix}.
\end{equation*}
Then, since the supercharges annihilate one bosonic state and one fermionic state, the Witten index is $$W= n_B^{E=0}-n_F^{E=0}=1-1=0,$$ and supersymmetry is not spontaneously broken.  Notably, these are solutions regardless of the particular choice of $h$.  On a line, the normalizability of these states depends on the form of $h$. If $h\to +\infty$ as $|q|\to \infty$, the fermionic wavefunction $f(q)$ is normalizable while $b(q)=0$ and the vacuum state is fermionic. If $h\to -\infty$ as $|q|\to \infty$, the bosonic wavefunction $b(q)$ is normalizable while $f(q)=0$ and the vacuum state is bosonic. Otherwise, there are no normalizable $E=0$ states and supersymmetry is spontaneously broken. However, since the circle has finite radius and $h$ is periodic, these states are \textit{always} normalizable on $S^1$. So although the invariance of the Witten index does not prevent this system from spontaneously breaking supersymmetry (as is does for scenarios such as the leftmost system in Figure \ref{fig:SUSY_shifting}), as the parameters of $h$ are varied, supersymmetry will always remain unbroken on a circle. 
\end{example}

\section{Discrete Morse Theory} \label{discrete Morse}

We have seen that the Witten index can be viewed as an alternating sum of  $\pm 1$ values over the critical points of a potential term and that the Witten index is an invariant of a supersymmetric system. In 1982, Witten established a connection between this property of supersymmetry in quantum mechanics and Morse theory \cite{Witten_1982}.  By thinking of the supercharges of a supersymmetric quantum mechanical system on a manifold as maps between spaces of differential forms, his perturbative description of Morse theory helps us to understand the geometry and topology of the underlying manifold. From the other direction, this connection shows that the invariant Witten index of a supersymmetric system can be understood as arising from the topology of the underlying manifold. This can then be translated into the discrete world. 

The following is a sketch of Witten's approach. For details see \cite{Witten_1982}. He begins by proposing an association of the Hilbert space of quantum states on a Riemannian manifold $(\mathcal{M}, g)$ with the de Rham Complex, a co-chain complex of differential forms on the manifold: $$(\Omega^\bullet (\mathcal{M}) , d):\;0\to\Omega^0 \left(\mathcal{M}\right) \xrightarrow{d}\Omega^1 \left(\mathcal{M}\right) \xrightarrow{d}\Omega^2\left(\mathcal{M}\right)\xrightarrow{d}\ldots\xrightarrow{d}\Omega^D\left(\mathcal{M}\right)\to 0,$$ where $d$ is the  de Rham operator. 

The metric Laplacian is given by \begin{equation} \label{eq:deRahm_Lap}
  \Delta_g=dd^\dagger+d^\dagger d.
\end{equation} 
Spin statistics are incorporated into this picture by associating even (odd) $n$-forms with bosons (fermions). For the supercharges $Q_1$ and $Q_2$ (Equation  \ref{eq: Q1_Q2_M}) \begin{equation*}
 Q_1=M+M^\dagger \;\; \text{and} \;\  Q_2=i(M-M^\dagger),
\end{equation*}
he makes the identification \begin{equation*}
Q_1=d+d^\dagger \;\; \text{and} \;\  Q_2=i(d-d^\dagger),\end{equation*} 
with \begin{equation*}
M\mapsto d \;\; \text{and} \;\; M ^\dagger\mapsto d^\dagger.\end{equation*}
The operators $$d_t=e^{-ht}de^{ht} \;\; \text{and} \;\; d^\dagger_t=e^{ht}d^\dagger e^{-ht},$$ for some Morse function $h$ and real parameter $t$,  also satisfy the supersymmetry algebras and have associated conjugate SUSY Hamiltonian $$H_t=\Delta_t=d_td_t^\dagger+d_t^\dagger d_t.$$
Then, for even $\lambda$, the Betti number $$b_\lambda(t)=\dim[\text{Ker}(d_t)]-\dim[\text{Im}(d_t)],$$ is the number of $p$-forms (quantum states) $\psi$ that obey $d_t\psi=0$ but cannot be written as $\psi=d_t\phi$ for any $p$-form $\phi$. This is the number zero-energy bosonic eigenstates of $H_t$. Moreover, for odd $\lambda$, $$b_\lambda(t)=\dim[\text{Ker}(d^\dagger_t)]-\dim[\text{Im}(d^\dagger_t)],$$
is the number of zero-energy fermionic eigenstates of $H_t$. 

In general, the spectrum of $H=dd^\dagger+d^\dagger d$ is challenging to find. But the Betti numbers associated with $d_t$ are the same as the Betti numbers associated with $d$ and therefore the number of zero-energy bosonic and fermionic eigenstates of $H$ is the same as that of $H_t$. 

So,  \begin{equation} \label{eq: Witten=Euler}
   W=n_B^{E=0}-n_F^{E=0}= \sum_{\lambda\; even} b_\lambda - \sum_{\lambda\;odd} b_\lambda = \chi (\mathcal{M}).
\end{equation}
Thus, the Witten index for a SUSY quantum system on $\mathcal{M}$ is equal to the Euler characteristic of $\mathcal{M}$.

\subsection{Discretizing Witten--Morse Theory}\label{s: MT_G}
A \textit{simplicial complex} is a set of points, edges with boundary points, faces with boundary edges, and their higher dimensional counterparts. A graph is therefore a 1--dimensional simplicial complex. Each individual vertex, edge, face, etc. that makes up the simplicial complex is called a \textit{simplex}. Forman demonstrates one method of constructing a combinatorial version of Morse theory on simplicial complexes which retains many of the intricate properties of continuous Morse theory \cite{Forman_2002, Forman_1998_2, Forman_1998}. In particular, he posits and proves the Morse inequalities for simplicial complexes, which look remarkably like the Morse inequalities on manifolds  \cite{Forman_1998_2}. In \cite{Forman_1998}, he then re-proves them from an approach mimicking Witten's for continuous Morse theory. The benefit of this result is that since the operators become finite dimensional when discretized, analytical issues of describing the deformed Laplacian in the continuum disappear. 
We outline the construction here, and refer the reader to \cite{Forman_2002, Forman_1998_2, Forman_1998} for details. 

A function on a simplicial complex assigns a value to each simplex. There is no obvious interpretation of critical simplices of such functions, and therefore no immediate classification of critical simplices as ``degenerate'' or ``non-degenerate,'' the characteristic we used to determine if a function on a manifold was a Morse function. Instead, Forman introduces a definition of a combinatorial Morse function, which we will see allows us to define a type of ``critical simplex'' with similar properties to non-degenerate critical points of functions on manifolds. 

Consider a simplicial complex. We say that a vertex is a $0$-dimensional simplex, an edge is a $1$-dimensional simplex, a face is a $2$-dimensional simplex, and so on. Each simplex of dimension $p$ is by definition bounded by a set of simplices of dimension $p-1$.  (On a graph, a vertex $u$ lies in the boundary of all of its incident edges and an edge $uv$ is bounded by its end vertices $u$ and $v$). Let $\sigma^{(p)}$ be a simplex with dimension $p$. For $\tau^{(p+1)}$, a simplex of dimension $p+1$, we say $\tau > \sigma$ if $\sigma$ lies in the boundary of $\tau$ and for $\tau^{(p-1)}$, a simplex of index $p-1$, we say $\tau<\sigma$ if $\tau$ lies in the boundary of $\sigma$. 

Let $S$ be a finite simplicial complex, $K$ \label{setsimp} the set of simplices of $S,$ and $K_p$ the subset of simplices of dimension $p$. A function $f: K\to\mathbb{R}$ is a \textit{discrete} Morse function if for every $\sigma^{(p)}\in K_p,$ \begin{equation} \label{eq:D_Morse_f}
    \begin{split}
        &\#\;\{ \tau^{(p+1)}>\sigma^{(p)}|f(\tau^{(p+1)})\leq f(\sigma^{(p)})\}\leq 1 \\
        &\#\;\{ \tau^{(p-1)}<\sigma^{(p)}|f(\tau^{(p-1)})\geq f(\sigma^{(p)})\}\leq 1.
    \end{split}
\end{equation}
On a graph $K(\Gamma)= V(\G)\oplus E(\G)$, this means that for every vertex $v\in V$, a discrete Morse function can assign no more than one incident edge a value less than it assigns to the vertex \textit{and} for every $e \in E$ the function can assign no more than one of its end vertices a value greater than it assigns to the edge. 

Then a given simplex $\sigma^{(p)}$ with dimension $p$ is critical if \begin{equation} \label{eq: critical simplex}
    \begin{split}
         &\#\;\{ \tau^{(p+1)}>\sigma|f(\tau)\leq f(\sigma)\}=0\\
        & \#\;\{ \tau^{(p-1)}<\sigma|f(\tau)\geq f(\sigma)\}=0 \;. 
    \end{split}
\end{equation}
On a graph, this means that a vertex is critical under a Morse function if the function assigns to it a smaller value than it does to each of its incident edges and an edge is critical if the function assigns to it a value greater than both of its end vertices. 

Figure \ref{fig:Graph Morse Function} shows two different functions on the cycle graph $C_3$. (A) is not a discrete Morse function, as the value assigned to  $v_1$ is greater than the value assigned to both of its incident edges and the value assigned to $e_3$ is less than the value assigned to its end vertices. On the other hand, (B) does satisfy the conditions for a Morse function, with $v_1$ a critical vertex and $e_3$ a critical edge.
\begin{figure}
    \centering
    \includegraphics[width=0.4\linewidth]{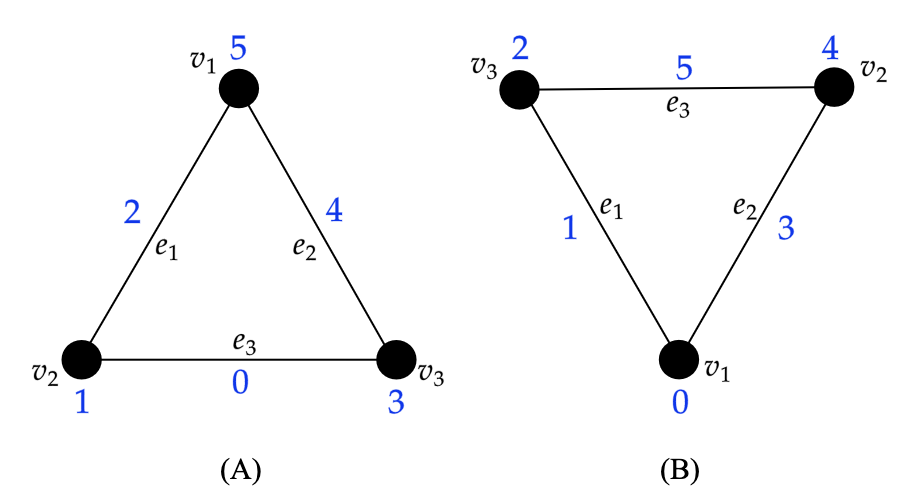}
    \caption{\small{A Morse function (blue) on a graph. Figure adapted from \textcite{Forman_2002}.}}
    \label{fig:Graph Morse Function}
\end{figure}

One  might ask what the index of a critical simplex is. As on a manifold, the index of a critical simplex is the dimension of the simplex. On a manifold, we determine the dimension of the cell based on the index of the critical point that determined it. This entails asking in how many independent directions an infinitesimal displacement from the critical point would result in a decrease in the output of the Morse function. Though the same question cannot be perfectly answered on a simplicial complex, we offer a lower-dimensional way to conceive of this same concept on graphs, in an attempt to make clearer the connection between the continuous Morse theory and Forman's discrete construction.

Let a small displacement from a critical simplex be a vertex-edge walk of length $1$ and let the ``direction'' of the displacement be the type of simplex where the walk terminates. Then, a displacement from a critical vertex moves in the ``edge direction,'' while one from a critical edge moves in the ``vertex'' direction. 

In \cite{Contreras_Xu_2019}, this construction is used in the context of discrete Morse theory for graphs, providing some insightful examples. Of particular relevance to our results of supersymmetry on graphs in are the following observations.

\begin{lemma}
    Non-critical graph simplices always come in adjacent vertex-edge pairs.
\end{lemma}
\begin{proof}
Let $e\in E(\G)$ be a critical edge of a Morse function $f$ on a graph $\G$. Let $v_1, v_2 \in V(\G)$ be the end vertices of $e$. Then, $f$ must take a greater value on exactly one of $v_1$ and $v_2$ than it does on $e$. That is, by definition of a critical simplex, either $f(v_1)>f(e)$ or $f(v_2)>f(e)$ and by definition of a discrete Morse function only one of these inequalities may hold and the other must be inverted. Without loss of generality, let $f(v_1)>f(e)$. Then, by Equation  \ref{eq: critical simplex}, $v_1$ is a critical vertex. If we begin instead with a non-critical vertex $v\in V(\G)$ of $f$, by the same logical sequence, exactly one of its incident edges must take a value smaller that $f(v)$, and thus that edge is also non-critical. 
\end{proof}
This observation indicates that the Morse inequalities will hold on a graph. Since the Euler characteristic of a graph is given by $\chi(\G)=|V|-|E|$, the ability to pair all non-critical simplices implies that that the value of the $\chi(\G)$ is determined completely by the critical simplices.

This lemma is used in \cite{Contreras_Xu_2019} to develop a discretization of gradient vectors and instantons on graphs, wherein the discrete gradient vector field $\nabla f$ of a combinatorial Morse function $f$ on $\G$ is the set of adjacent non-critical simplex pairs. This results in a similar decomposition with which we began our study of Morse theory.

Then, from a graph, we construct the chain complex  \begin{equation} \label{eq: graph_chain_cmplx}
    0\to\mathbb{C}^{|V|}\xrightarrow{\partial}\mathbb{C}^{|E|}\to 0,
\end{equation}
where the vertex states $\mathbb{C}^{|V|}$ correspond to the zero-forms (Witten's bosonic states) and the edge states $\mathbb{C}^{|E|}$ correspond to the 1-forms (Witten's fermionic states). To find the boundary operator, we can ask which operators act from vertex states to edge states and and vice versa. But we have already seen such an operator, the incidence matrix!

 Indeed, we can understand the graph incidence operator and its transpose as discrete versions of the de Rham co-boundary and boundary operators:\begin{equation}
     d=I^T \;\;\; \text{and}\;\;\; d^\dagger=I.
 \end{equation}
From the supersymmetric picture, we would say that the graph supercharges are of the form \begin{equation}
    Q_\G =  \left (\begin{array}{c|c}
    \  0 &  I\ \\ \hline
    \ I^T & 0 \ \\  \end{array} \right ),
\end{equation} with \begin{equation}
    M \mapsto I^T \;\; \;\text{and} \;\;\; M^\dagger \mapsto I.
\end{equation}
Then, the induced graph Laplace operator is given by \begin{equation} \label{eq: SUSY_G_L}
    \Delta_{S,\G}=dd^\dagger+d^\dagger d = II^T + I^TI = \Delta_+ \oplus \Delta_-. 
\end{equation} We call this direct sum of the even (bosonic) graph Laplacian and the odd (fermionic) graph Laplacian the \textit{supersymmetric graph Laplacian}. 

The combinatorial supercharges constructed with the graph incidence matrix obey the supersymmetry algebras, with the exception that $I^2\neq {I^T}^2 \neq 0.$ (Indeed, in most cases this equation doesn't even make sense, since $I$ and $I^T$ are not usually square matrices.) However, they are still appropriate choices for $d^\dagger$ and $d$ because the truncation for the graph cochain complex trivially mimics this effect. 

The recovery of our even and odd graph Laplacian when we consider Morse theory on graphs is ideal, because we have already shown that there is a close relationship between the steady states of the graph Lapalacian and the topology of the graph. Specifically, we proved that
\begin{enumerate}
    \item $\dim[\ker(\Delta_+ (\G))]=\# \; \text{of connected components of}\;\G$ (Theorem \ref{Thm: DimKer+})
    \item $\dim[\ker(\Delta_- (\G))]=\# \; \text{of independent cycles of}\;\G$ (Theorem \ref{Thm: dimker-}). 
\end{enumerate} 

We can use this knowledge to find the Betti numbers of the graph.
\begin{theorem}
    On a graph, $b_0$ and $b_1$ equal the number of connected components and independent cycles, respectively. 
\end{theorem}\begin{proof}
 Let $\G$ be a finite graph. Then $$b_p=\dim[H_p(\G)]=\dim[\text{ker}(\partial_p)/\text{im}(\partial_{p+1})]=\dim[\text{ker}(\partial_p)]-\dim[\text{im}(\partial_{p+1})].$$ But $\dim [\text{im}(\partial_{p+1})]=1$, since the chain is truncated and everything is sent to $0$. Thus, \begin{equation*}
    \begin{split}
        b_0&= \dim[H_v(K)]=\dim[\text{ker}(I^T)]\\ b_1&=\dim[H_e(K)]=\dim[\text{ker}(I)] \;\; .
    \end{split}
\end{equation*}
On the other hand, from Theorems \ref{Thm:Ker+} and \ref{Thm:Ker-} that $\dim[\text{ker}(I^T)] = \# \;\text{of connected components  of}\; \G$ and $\dim[\text{ker}(I)]= \# \;\text{of independent cycles of}\; \G $. 
This completes the proof. 
\end{proof}

We follow this with the observation that, by our analysis of the zero-energy states of the even and odd graph Laplacians in Section \ref{s:SteadSF_G}, the vertex Betti number $b_0$ is exactly the number of bosonic zero-energy eigenstates and the Betti number $b_1$ the number of fermionic zero-energy eigenstates. 

Finally, this let's us see that the Euler characteristic of a graph is equal to the Witten index of a graph for our supersymmetric graph Laplacian. 

\begin{equation}\begin{split} \label{eq: G_Witten_Index}
      \chi(\G)&=|V|-|E|\\&=b_0-b_1\\&= \#\text{connected components} - \#\text{independent cycles}\\&=\#\text{$0$-energy vertex (bosonic) states} - \#\text{$0$-energy edge (fermionic) states} \\&=W(\G).
\end{split}
\end{equation}

This is precisely the relation Witten exploited on manifolds to prove the Morse inequalities using supersymmetry. Knowing that this relation holds on graphs, we show how we can use it to study combinatorial supersymmetric quantum mechanics. 

\section{Main Results} \label{Results}
In the previous section, we worked our way along the black arrows, from supersymmetry to combinatorial Morse theory using Witten's perturbative approach. The link between supersymmetry and Morse theory in the continuum and the observation that Morse theory can be implemented on graphs, leads us to believe that realizing supersymmetric quantum mechanics on graphs is possible and may yield interesting results, both in the study of quantum theory on graphs, and in the study of graphs themselves.

Indeed, we now show that graphs are, mathematically, naturally supersymmetric systems. Combining the framework of discrete Morse theory our graphical model of quantum mechanics, we posit and prove some new theorems that enable us to begin to untangle the supersymmetric nature of graphs and to use graphs to model supersymmetric quantum mechanical systems. These results concern the spectrum of the supersymmetric graph Laplacian (Section \ref{s: SUSY_L_spectrum}), the Dirac operator on graphs and supersymmetric path integrals (Section \ref{s:SUSY_Dirac_G}), and spontaneous supersymmetry breaking (Sections \ref{s: SSB_G} and \ref{s:rewiring}). 

\subsection{The Supersymmetric Structure of Graphs} \label{s: SUSY_L_spectrum} \nocite{Ogurisu_2003} \nocite{Ogurisu_2002} \nocite{Post_2009} \nocite{Requardt_2005}
We have seen that graphs display a natural mathematically supersymmetric structure, with the incidence matrix and its transpose forming ``supercharge''-like objects which act between vertex and edge states. Drawing a vertex-to-boson and edge-to-fermion analogy allowed us to construct what we called the supersymmetric graph Laplacian (Equation \ref{eq: SUSY_G_L}) which we used to find the Betti numbers of the graph. But, does the association of vertex states as bosonic and edge states actually make \textit{physical} sense? The answer, we argue, is yes.

The Hilbert space of quantum mechanical states can be decomposed into a direct sum of bosonic and fermionic sectors. We should be able to describe the evolution of states within each sector, independently, and the direct sum of these bosonic and fermionic evolution operators should describe the evolution of mixed states. Then, given our supersymmetric graph Laplacian\begin{equation*}
    \Delta_{S,\G}= II^T + I^TI, 
\end{equation*} we would expect $II^T$ and $I^TI$  to act as Laplace operators in the vertex and edge sectors, respectively.

We have already shown, when constructing quantum mechanics on graph that they do precisely this! The even graph Laplacian $\Delta_+=II^T$ controls the evolution vertex states $|\psi\ket_v$ via the vertex Schr\"odinger equation (Equation \ref{eq:Schrod+_G}) and the odd graph Laplacian $\Delta_-=I^TI$ controls the evolution of edge states $|\psi\ket_e$ via the edge Schr\"odinger (Equation \ref{eq:Schrod-_G}). This indicates to us that the association is, at least at face-value, sound. 

Let us now lay out clearly the players in our graph construction of supersymmetric quantum mechanics. In the continuum, we considered the Hilbert space $$\mathcal{H}=\mathcal{H}^+\oplus \mathcal{H^-}= L^2(\mathbb{C}),$$
with bosonic sector $\mathcal{H}^+$ and fermionic sector $\mathcal{H}^-$. On graphs, we consider the analogous Hilbert space of normalizable functions which assign a complex value to each simplex $$\mathcal{H}_\G=L^2(\G)=\left\{ |\psi\ket: K (\G)\to \mathbb{C}, \sum_{\sigma \in S(G)} |f(\sigma)|^2<\infty \right\}.$$
Since the graph partitions neatly into vertex and edge sets $K(\G)=V(\G)\oplus E(\G)$, 
$$\mathcal{H}_\G=\mathcal{H}_\G^+\oplus\mathcal{H}_\G^-=L^2(\G)=L^2(V(\G))\oplus L^2(E(\G)),$$ with vertex-bosonic sector $\mathcal{H}_\G^+$ and edge-fermionic sector $\mathcal{H}_\G^-$.\footnote{We will sometimes refer to these sectors (and the states within them) by only either the simplex type or the particle type, whichever is most relevant, and the other should be understood to be true also.}

We consider the evolution of mixed states in $\mathcal{H_\G}$, $$ |\psi\ket_{v,e}= \begin{pmatrix}
      B_v\\ \hline F_e
    \end{pmatrix}  ,$$ governed by the supersymmetric Laplacian (the free SUSY Hamiltonian) $$ \Delta_{S}= \begin{pmatrix}
        \Delta_+&0\\0&\Delta_-
    \end{pmatrix}.$$
Graph supercharges, which act between $\mathcal{H}_\G^+$ and $\mathcal{H}_\G^-$, take the form \begin{equation*}
    Q_{\G 1} =  \left (\begin{array}{c|c}
    \  0 &  I\ \\ \hline
    \ I^T & 0 \ \\  \end{array} \right ), \;\;\;  Q_{\G 2} =  \left (\begin{array}{c|c}
    \  0 &  iI\ \\ \hline
    \ -iI^T & 0 \ \\  \end{array} \right ), 
\end{equation*}
where $I$ is the incidence operator of $\G$. Notice that the continuum superalgebra structure is preserved, with $$[Q_\G, \Delta_{S}]=0 \;\; \text{and} \;\; Q_\G^2=\Delta_{S}.$$

We propose a natural graph analog of the fermion parity operator. Let $\G$ be a graph with order $m$ and size $n$. Then, the $(m+n)\times(m+n)$ diagonal matrix \label{gferp} \begin{equation}
    F_\G=  \left (\begin{array}{c|c}
    \ \;1 & 0\;\\ \hline
     \ \;0 & -1 \;
 \end{array} \right ) \begin{matrix}
   \} \; \text{\tiny{$m\times m $}} \\ \}  \; \text{\tiny{$n\times n$}} \;\;
    \end{matrix}, 
\end{equation} with the property that $$\{Q_i, F_\G\}=0,$$ acts as a fermion parity operator on graphs, distinguishing vertex states from edge states by 
 \begin{equation*}
    F_\G|\psi\ket_{v,e}= \left (\begin{array}{c|c}
    \ \;1 & 0\;\\ \hline
     \ \;0 & -1 \;
 \end{array} \right ) \begin{pmatrix}
      B_v\\ \hline F_e
    \end{pmatrix} = \begin{pmatrix}
      B_v\\ \hline -F_e
    \end{pmatrix}.
\end{equation*}
Recall that in the continuum, the (regulated) trace of the fermion parity operator is equivalent to the Witten index (Equation \ref{eq: W_Index_Tr}). Likewise, the trace of $F_\G$ is the Euler characteristic of the graph, which we have shown to be equivalent to the Witten Index of a SUSY graph quantum system: 
$$\text{Tr}(F_\G)=\sum m(1)+n(-1)= |V|-|E|= \chi(\G)= W(\G).$$
We can say more by examining the spectrum of the supersymmetric graph Laplacian.  

\begin{theorem}
\label{thm:SUSY_degeneracy}All non-zero eigenvalues of the supersymmetric graph Laplacian $\Delta_S=\Delta_+\oplus\Delta_-$  are degenerate.
\end{theorem}
\begin{proof}
Recall that that the spectra of the even and odd graph Laplacians $\Delta_+$ and $\Delta_-$ coincide in their non-zero eigenvalues (Theorem \ref{thm:+-_same_spectrum}). Then, since the SUSY graph Laplacian is the direct sum of the even and odd graph Laplacians, the spectrum of $\Delta_S$ is the union of their spectra and each non-zero eigenvalue of $\Delta_S$ has at least twofold degeneracy.
\end{proof}

\begin{theorem}
    The supersymmetric graph Laplacian is positive semi-definite.
\end{theorem}
\begin{proof}
The spectra of both the even graph Laplacian $\Delta_+$ and the odd graph Laplacian $\Delta_-$ are positive semi-definite. This is apparent from the quadratic form of the even Laplacian (Equation \ref{eq:+_qform}) and the coincidence of the non-zero eigenvalues of the two graph Laplacians (\ref{thm:+-_same_spectrum}). Then since $\Delta_S = \Delta_+ \oplus \Delta_- $, the spectrum of $\Delta_S$ is positive semi-definite.
\end{proof}
Thus, on a graph, the non-zero energies of a graph supersymmetric system are degenerate and positive, as is  true for supersymmetric systems in the continuum (Theorems \ref{thm: D_SUSY_H} and \ref{thm: PD_SUSY_H} ).

\subsection{Walk Sums and the Graph Dirac Operator, Revisited} \label{s:SUSY_Dirac_G} \nocite{Mnev_2007} \nocite{Yu_2017}
In the continuum, the square of the Dirac operator $\slashed{\partial}^2$ can be considered as the Hamiltonian of a supersymmetric quantum mechanical system. Cooper et al. give a clear justification of this relationship in \cite{Cooper_Khare_Musto_Wipf_1988}. In particular, they study the spectrum and eigenfunctions of $\slashed{\partial}^2$ by exploiting SUSY methods. They discuss two supersymmetric decompositions of $\slashed{\partial}^2.$ We use the first of these, called the chiral supersymmetry. Using this decomposition, the Dirac operator can be shown to be related to the supercharges by
\begin{equation*}\begin{split}
   \slashed{\partial}&=M+M^\dagger \\\slashed{\partial}^2&= H_S= \left (\begin{array}{c|c}
    \  H^+ &  0\\ \hline
     \ 0 & H^-  \\ 
 \end{array} \right )=\left (\begin{array}{c|c}
    \  M^\dagger M &  0\\ \hline
     \ 0 & M M^\dagger \\ 
 \end{array} \right ).
\end{split}
\end{equation*}

By design, our discrete Dirac operator $\slashed{D}_I$ and supersymmetric graph Laplacian preserve this relationship. That is, in Section \ref{graph QM}, after defining vertex-edge states on graphs, we constructed the incidence Dirac operator (Equation \ref{eq:D_G}) as $$\slashed{D}_I=\begin{pmatrix}
    0&I\\I^T&0
\end{pmatrix}.$$This was a square root of the direct sum of the even and odd graph Laplacians (Equation \ref{eq:D^2_G}) $$\slashed {D}_I^2=\begin{pmatrix}
    \Delta_+&0\\0&\Delta_-
\end{pmatrix},$$
which we know recognize as the supersymmetric graph Laplacian, or the free supersymmetric Hamiltonian.

In the continuum, recognition of the Dirac operator as a SUSY Hamiltonian enables the simpler analytic methods of solving the Schr\"odinger equation to be leveraged toward solving the Dirac equation, with the supersymmetric Hamiltonian as the intermediary. We propose that studying the incidence Dirac operator as a discrete SUSY operator in the context of walk sums could allow for even further simplification of the analytic methods via finitely convergent sums on graphs.  

Recall that the $ij$ elements of the $k^{th}$ power of the graph Dirac operator counts the number of vertex-edge walks of length $k$ of non-canceling sign (Equation \ref{eq:D_k_walk}):$$[\slashed {D}_I]^k_{ij}=\sum_{\gamma: \;i\to j, \;k} \text{sgn}(\gamma_{ve}).$$
Since $\slashed{D}_I^2$ is a SUSY graph Hamiltonian, we can consider the graph Euclidean propagator \begin{equation*}
    \bra\sigma_i|e^{t\slashed{D}^2}|\sigma_j\ket=\sum_{k=0}^\infty \frac{t^k}{k!}(\slashed{D}^2)^k,
\end{equation*}
where $\sigma_i$ and $\sigma_j$ are two simplices of the graphs. We can understand this as a sum over vertex-edge walks of length $2k$: 
\begin{equation*}
    \bra\sigma_i|e^{t\slashed{D}^2}|\sigma_j\ket=\sum_{\text{walks}\ \gamma_{\sigma_i\sigma_j}} \text{sgn} (\gamma_{ve})\frac{t^{(|\gamma|/2)}}{(|\gamma|/2)!}
\end{equation*}
Recall that vertex-edge walks of even length are really just vertex and edge superwalks.

We can also compute the trace of the propagator, summing over all closed walks. $$\bra\sigma|e^{t\slashed{D}^2}|\sigma\ket=\sum_{k=0}^\infty \frac{t^k}{k!} \ \left (\sum_{\gamma: \sigma\to\sigma,\ 2}\text{sgn}(\gamma_{ve}) \right )^k.$$
Any $\gamma_{ve}$ a length--$2$ vertex-edge walk can be equivalently though of as a length-$1$ superwalk $\gamma'$. For both vertex superwalks and edge superwalks, if the superwalk is hesitant (returns to the same simplice $\sigma$ where it began), $\text{sgn}(\gamma')=1$.  Furthermore, if $\sigma$ is a vertex there are $\text{deg}(\sigma)$ unique choices of walk (one through each incident edge) and if $\sigma$ is an edge, there are exactly two unique walks (one through each end vertex). So, $$\bra v_i |e^{t\slashed{D}^2}|v_i\ket=\sum_{k=0}^\infty \frac{t^k}{k!} (\text{deg(v)})^k$$ and 
  $$\bra e_i |e^{t\slashed{D}^2}|e_i\ket=\sum_{k=0}^\infty \frac{t^k}{k!} (2)^k.$$

\subsection{Spontaneous Supersymmetry Breaking on Graphs} \label{s: SSB_G}
We've seen that the continuum, it can be quite challenging to determine conclusively if the vacuum energy of a supersymmetric Hamiltonian is positive, thus breaking supersymmetry. However, a Witten index of zero is a \textit{necessary }condition for supersymmetry breaking and is far more easily computed, as it is a topological invariant of the underlying manifold. We have also seen via discrete Witten-Morse theory that the Witten index of a graph SUSY system is a topological invariant of the graph. More precisely, \begin{equation} \label{eq:WI_G}
    W_\G= \text{Tr}(F_\G)= |V|-|E|= \# \text{connected components} - \# \text{independent cycles}.
\end{equation}
Under the reasonable proposition that supersymmetry on a graph is, as in the continuum, said to be spontaneously broken when there are no zero-energy vertex-bosonic or edge-fermionic states, $W_\G=0$ is likewise a necessary condition  for spontaneous supersymmetry breaking on a graph. We can immediately see from Equation \ref{eq:WI_G} that this condition can only be satisfied only on the particular class of graphs for which $|V|=|E|$.\footnote{By Equation \ref{eq:WI_G}, $|V|=|E|\Leftrightarrow$ the number of connected components of $\G$ equals the number of independent cycles of $\G$.}

Interestingly, this condition expressly does not hold on a lattice graph, the graph proposed by Mnev in \cite{Mnev_2016} to recover non-supersymmetric quantum mechanics in the continuum when embedded on a torus. Thus, in our construction, without additional impositions on the graph structure, a ``realistic'' lattice-graph model of quantum theory is incompatible with a ``realistic'' supersymmetric one.\footnote{Motivated by a desire to link lattice QCD with superstring quantum gravity theories, supersymmetry on a lattice is an active field of study that has puzzled physicists for a long time. Naively discretizing continuum QCD theories and hoping to stumble upon an action on the lattice that does not explicitly break supersymmetry is nearly impossible. However, Catterall has demonstrated that for certain chiral supersymmetries that are unbroken in the continuum but explicitly broken at non-zero lattice spacings, the implementation of ``topological twisting'' can preserve the supersymmetry on the lattice \cite{Catterall_2015}. Once a lattice supersymmetry is found, the question then becomes whether it can be spontaneously broken. Catterall and Veernala found that in their topologically twisted lattice super QCD, the supersymmetry can break spontaneously \cite{Catterall_Veernala_2015}. A graph theoretic analysis of such a twisting mechanism from the perspective of the graph supersymmetric structure developed in this work would be an interesting path to explore further.} 

Just as we can, in the continuum, rule out supersymmetric theories that have zero-energy vacuum states in a finite volume and thus are not candidates spontaneously broken SUSY in the infinite volume limit, we can consider our (finite) graph to act like a finite volume quantum system where we can study supersymmetry more simply. We present an example of this, returning to the case of a particle on $S^1$ seen previously in the continuum in Example \ref{ex:Particle_on_circle}. \newline

\begin{example}   \label{ex:particle_on_a_cycle}  \textbf{Particle on an \textit{n}-Cycle}\\
Consider the SUSY Hilbert space of vertex-edge states on the cycle graph $C_n$ $$\mathcal{H}_\G=\mathcal{H}_\G^+\oplus\mathcal{H}_\G^-=L^2(C_n)=\left\{ |\psi\ket: K (C_n)\to \mathbb{C}, \sum_{\sigma \in S(G)} |f(\sigma)|^2<\infty \right\}.$$

We can see immediately that, for any system in the space, $W=0$ but supersymmetry is not spontaneously broken since the graph is connected and has one cycle. This is exactly as we saw in the case of the supersymmetric particle on a circle in Example \ref{ex:Particle_on_circle}.

Furthermore, the supersymmetric Laplacian of a cycle graph is \begin{equation}
    \Delta_S=\Delta_+\oplus\Delta_-=\begin{pmatrix}
          2&-1& 0&\dots&-1\\-1&2&-1&\dots& 0\\0&-1&2&\dots&0\\\vdots&\vdots&\ddots&\ddots&\vdots\\-1&0&\dots &-1&2
    \end{pmatrix} \oplus \Delta_-.
\end{equation}
Since $\Delta_+$ is orientation independent, we can write it out explicitly, though we can't do the same for the orientation dependent $\Delta_+$. Notice, that $\Delta_+$ is circulant and each of its diagonal entries is $2$ since $C_n$ is $2$-regular.\footnote{By contrast, $\Delta_-$ need not be circulant. For example, consider $C_3$ with two edges oriented clockwise and one oriented counter-clockwise.} Then, we may express $\Delta_+$ as the sum $$\Delta_+=D-A=\begin{pmatrix}
          2&0& 0&\dots&0\\0&2&0&\dots& 0\\0&0&2&\dots&0\\\vdots&\vdots&\ddots&\ddots&\vdots\\0&0&\dots &0&2
    \end{pmatrix}-\begin{pmatrix}
          0&1& 0&\dots&1\\1&0&1&\dots& 0\\0&1&0&\dots&0\\\vdots&\vdots&\ddots&\ddots&\vdots\\1&0&\dots &1&0
    \end{pmatrix}.$$
 Evidently, the spectrum of $D$ is $\{2\}$, with multiplicity $n$. The adjacency matrix $A$ is also circulant, so the form of its spectrum is well known (see, for example \cite{Lee_Luo_Sagan_Yeh_1992}). Its eigenvalues are \begin{equation*}
      \lambda_{(k+1)}=x^k+x^{(n-1)k}=e^{2k\pi i/n}+e^{-2k\pi i /n}=2\cos(2k\pi/n),
 \end{equation*} for $k=0,1,\dots(n-1).$ Then, the eigenvalues of $\Delta_+$ are \begin{equation*} \begin{split}
      \mu_{(k+1)}=2-2\cos(2k\pi/n)=4\sin^2(k\pi/n), \;\;\; k=0,1,\dots(n-1).
 \end{split}
 \end{equation*}
 If we wish to take a continuum limit of the spectrum, we impose a metric on the graph such that each edge has a weight $h=2\pi/n$ and consider the metric Laplacian $\frac{1}{h^2}\Delta_+.$ This can be thought of as describing a uniform physical spacing between points of the cycle. The spectrum becomes
$$\mu_{k+1}= \frac{1}{\pi^2}\left(n^2\sin^2(k\pi/n)\right ) = \frac{1}{\pi^2} \left( \frac{\sin(k\pi/n)}{(1/n)} \right)^2 $$
 
Then, we increase the number of vertices and edges of our graph (taking a more dense discrete sampling). In the limit of $n\to\infty$ and $h\to0,$

$$\mu_{k+1}\to  \frac{1}{\pi^2}(k\pi)^2=k^2.$$
 This is exactly the spectrum of a free particle on $S^1$.

 Note that by Theorem \ref{thm:+-_same_spectrum}, the odd Laplacian has the same spectrum as the even Laplacian, so we can guarantee the recovery of bosonic-fermionic state degeneracy in the continuum limit of the spectrum of $\Delta_S$.  However, it is not yet clear to us exactly how to impose a metric on the odd graph Laplacian nor what the physical meaning of the orientations of edges and how they behave in the continuum limit\footnote{We suggest naively that the orientation of edges may help us build a discrete notion of curl, for the view of the graph Laplacian as a operator on $1$-forms.}.
\end{example}

We've seen that in the case of a lattice graph, a non-zero Witten index prevents supersymmetry from being broken, while in the case of a cycle graph, supersymmetry is unbroken despite that $W(C_n)=0$. For graphs with $W(\G)=0$,  like $C_n$,  we can strengthen the conditions for spontaneous supersymmetry breaking. 

\begin{theorem} \label{thm:SUSY_no_breaking} For the free SUSY Hamiltonian on a combinatorial graph, supersymmetry is spontaneously broken if and only if the graph is null.
\end{theorem}
\begin{proof}
($\Rightarrow$) If supersymmetry is spontaneously broken, then there must be no zero-energy vertex-bosonic states and no zero-energy edge-fermionic states. Recall that the zero-energy (steady state) solutions of the free graph Schr\"odinger equations are determined by the kernels of the even and odd graph Laplacians. By Theorem \ref{Thm: DimKer+}, the number of zero-energy bosonic (vertex) states is equal to the number of connected components of the graph. Then, the graph must contain no connected components, ie. the graph must be null. 

($ \Leftarrow$) If the graph is null, there are neither vertices nor edges, so the even and odd Laplacian matrices have no eigenvalues. Then, there are no zero-energy fermionic or bosonic states (or indeed, any states at all). 
\end{proof}
\textit{Remark:} The lack of combinatorial graphs on which spontaneous supersymmetry breaking can arise for a free system fits with the physical intuition. In a free theory, fields are non-interacting, but in a supersymmetric field theory, bosonic and fermionic fields necessarily interact, since supercharge operators act between $\mathcal{H^+}$ and $\mathcal{H^-}.$ Thus, without a potential term, the Lagrangian for a free theory lacks the necessary mathematical complexity for supersymmetry to properly manifest.

\subsection{Rewiring and the Spectrum of the SUSY Graph Laplacian  } \label{s:rewiring}
Although supersymmetry cannot be spontaneously broken for the free graph Hamiltonian on combinatorial graphs, we investigate in this section how varying what we will call the \textit{wiring parameter} of such graphs changes the spectrum of the supersymmetric graph Laplacian. We will see that the supersymmetric nature of graphs holds under the variation of this parameter. 

We define the wiring parameter of a finite graph $\G$ to specify which vertices in $V(\G)$ bound which edges in $E(\G).$ Then, we may vary this parameter by \textit{rewiring} the graph, or swapping the end vertices of the edges. If wiring is to behave as a parameter of a SUSY system, rewirings must preserve the dimension of $\mathcal{H}_\G$ as well as the dimensions of each of the vertex and edge sectors, so no edge or vertex may be removed from or added to the graph in the rewiring process. Since we are working on simple graphs, no edge $e=v_iv_j$ can be rewired such that both of its end vertices are the same $(v_i=v_j)$, nor can it be rewired to share both end vertices with another edge. We propose a notion of ``smooth rewirings'' wherein only one edge may be moved at a time, so mimic the smooth variation of continuous parameters. 

\begin{theorem} \label{thm: rwiring_1}
    The Witten index of a graph is invariant under rewiring.
\end{theorem}
\begin{proof}
   Since a rewiring cannot not add or remove simplices of the graph,  $W(\G)= |V(\G)|-|E(\G)|$ is invariant.
\end{proof}
We can use the invariance of the Witten index under rewiring to constrain the possible rewired configurations of any graph. We prove a series of lemmas.

\begin{lemma}
\label{lem:rewiring_1} A rewiring of a graph increases (decreases) the number of connected components of a graph if and only if it increases (decreases) the number of independent cycles by the same number.
\end{lemma}
\begin{proof}
Since we may equivalently define the Witten index to be $$W (\G)= \# \text{connected components} - \# \text{independent cycles},$$ by the invariance of the Witten index, any change in the number of connected components must be counterbalanced by an equivalent change in the number of independent cycles. 
\end{proof}
In the context of physical systems on graphs, this means that eigenstates that are raised from the zero-energy or lowered to zero-energy under variation in the wiring parameter always shift in boson-fermion pairs. Indeed this could have been anticipated also from the degeneracy of the spectrum of the SUSY graph Laplacian. 

While the proof of Lemma \ref{lem:rewiring_1} is brief, without a thorough understanding of the Witten index of graphs, it is not at all evident from simply looking at a graph that this would would be the case. We provide a few illustrative examples. 

\begin{example}
Figure \ref{fig:Rewiring_5} shows all possible rewirings of the cycle graph $C_6$ with $6$ vertices and $6$ edges, and $W(\G)=0$. Independent cycles are highlighted. The first row contains rewirings with one cycle, the second row rewirings with two cycles, and the third row rewirings with three cycles. As the graph is rewired to form an additional cycle, it also gains an additional connected component. While the number of vacuum states increases, the Witten index remains constant. 
\begin{figure}[ht]
    \centering
    \includegraphics[width=0.6\linewidth]{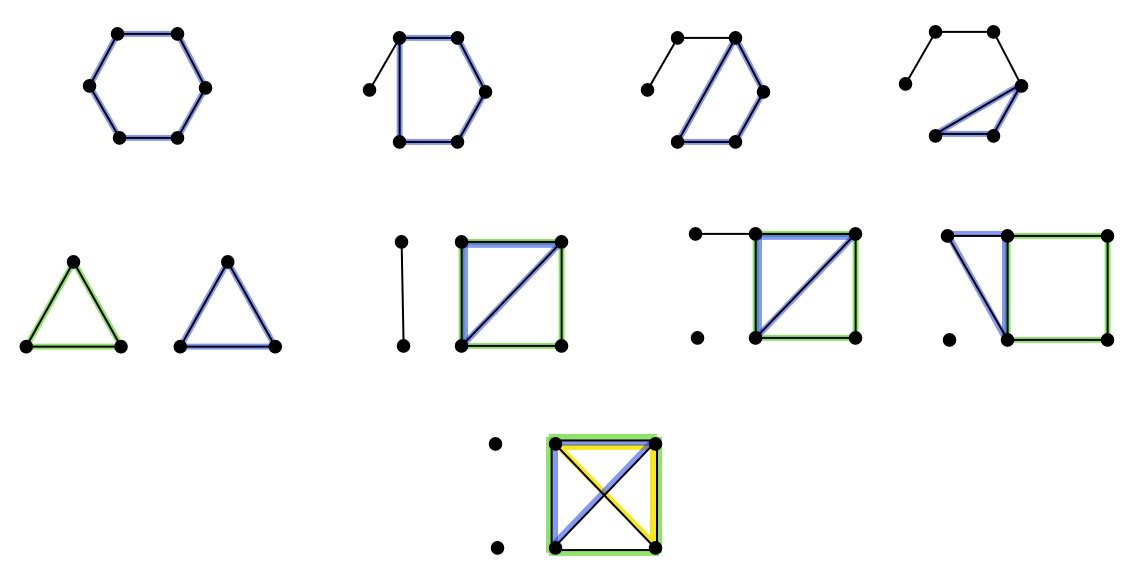}
    \caption{\small{All possible rewirings of $C_6$ with highlighted independent cycles. In each, the number of connected components equals the number of independent cycles.}}
    \label{fig:Rewiring_5}
\end{figure}
\end{example}

\begin{example}
If removing a particular edge of a graph disconnects a connected component (that is, it splits a connected component into two connected components), we will call that edge a \textit{bottleneck.} The bottlenecks of the large connected graph in Figure \ref{fig:Rewiring_2} are colored in red. It is challenging, however, on such large graphs to see that \textit{any} rewiring of one of these bottlenecks which disconnects the graph \textit{necessarily} forms a new independent cycle of the graph. But this is exactly what Lemma \ref{lem:rewiring_1} tells us. 

For the smaller, more manageable graph Hilbert space shown in Figure \ref{fig:Rewiring_3}, we can clearly see that when a bottleneck (red) is rewired to disconnect the top graph, a new independent cycle is formed in the resultant graph. Under these rewirings, the number of vacuum states decreases, but the Witten index, again, remains constant. 

\begin{figure}[ht]
    \centering
    \includegraphics[width=0.4\linewidth]{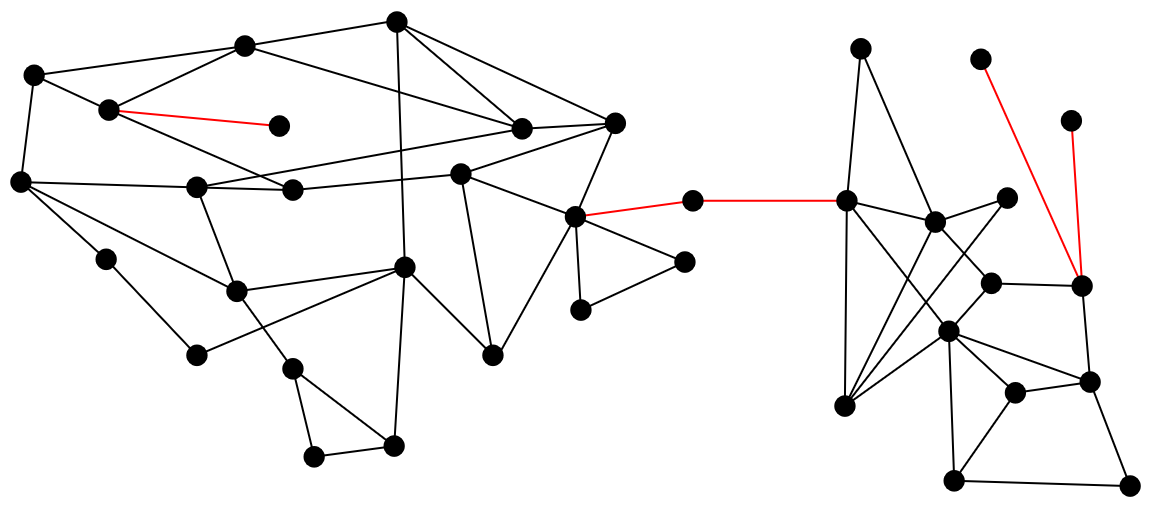}
    \caption{\small{A large connected graph with bottlenecks shown in red.}}
    \label{fig:Rewiring_2}
\end{figure}
    \begin{figure}[ht]
    \centering
    \includegraphics[width=0.5\linewidth]{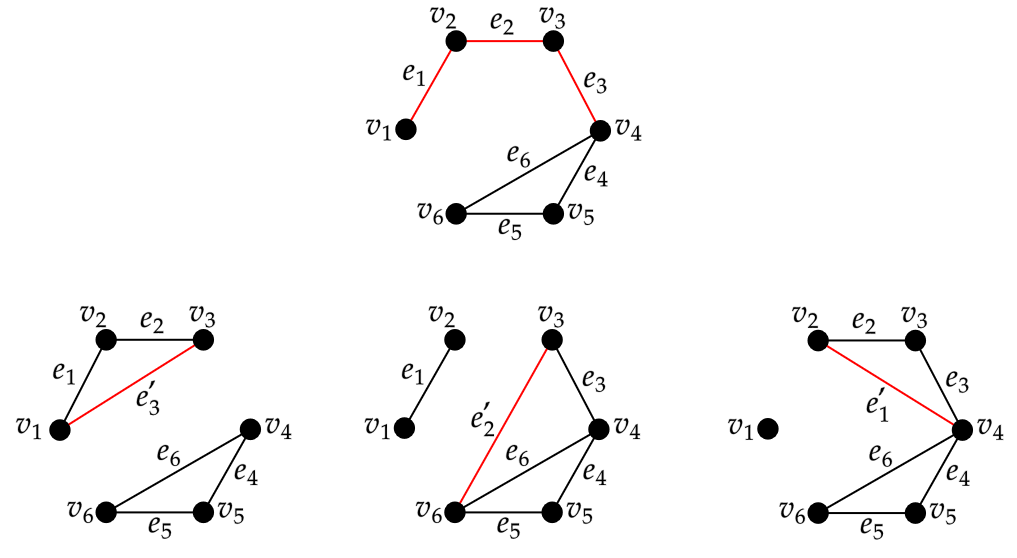}
    \caption{\small{Rewiring a bottleneck (red) to disconnect a connected graph (top line) adds an independent cycle to the resultant graph (bottom line).}}
    \label{fig:Rewiring_3}
\end{figure}
\end{example}
Based on the sign of the Witten index, we can describe the minimally populated vacuum state for supersymmetric graph Hilbert spaces under rewiring. 
\begin{lemma}
Let $n_{CC}$ and $n_{IC}$ be the number of connected components and the number of independent cycles of a graph Hilbert space, respectively. \\ \textit{Case 1:} If  $W>0$, then $n_{IC}< n_{CC}$ and smooth rewiring can reduce $n_{IC}$ to $0$, producing a fully bosonic vacuum.\\ \textit{Case 2:} If $W\geq0,$ then $n_{CC}\leq n_{IC}$ and smooth rewiring can reduce $n_{IC}$ by $n_{CC}-1.$ Then, if (a) If $n_{CC}= n_{IC}$ rewiring can produce a graph with $n_{CC}=n_{IC}=1$ (one fermionic and one bosonic vacuum state) or (b) if $n_{CC} < n_{IC}$ rewiring can reduce the number of fermionic states until there is only one remaining connected component (fermion-dominated vacuum state). 
\end{lemma}
\begin{proof}
Case 1:  $n_{IC}< n_{CC}.$ \\On a connected graph, $n_{CC}=1,$ so $n_{IC}=0,$ and we are done. Suppose $n_{CC}>1$ and $n_{IC}\neq 0.$ Then, for each independent cycle in $\G$, there must also be at least one connected component containing no cycles. We may label the independent cycles of $\G$ $c_1,\dots,c_{n_{IC}}$ and the cycle-less connected components of $\G$ $k_1,\dots k_{n_{CC}}.$ Notice that the set of independent cycles can have no vertices in common with the set of cycle-less connected components. Then, without loss of generality, one edge of $c_1$ can be rewired such that one of its end vertices lies in $k_1.$ This rewiring removes the cycle $c_1$ and connects the component initially containing $c_1$ with $k_1,$ satisfying Lemma \ref{lem:rewiring_1}. Likewise, one edge  of $c_2$ can be rewired such that one of its end vertices lies in $k_2,$ one edge of $c_3$ can be rewired such that one of its end vertices lies in $k_3,$ and so on, until all of the cycles are removed. 

Case 2: $n_{CC}\leq n_{IC}.$ \\ Let $c_1,\dots,c_{n_{IC}}$ be the independent cycles of $\G$ and let $k_1,\dots k_{n_{CC}}$ be the (not necessarily cycle-less) connected components of $\G.$ An edge in each $c_i$ may be successively rewired to remove $c_i$, connecting the component of $\G$ initially containing $c_i,$ with a different connected component $k_j$. This may be done until the graph is connected. At this point, $n_{CC}=1$ and cannot be decreased further. So, $n_{IC}$ can be reduced, at most, by $n_{CC}-1.$ 
\end{proof}
Notice that Case 2 ensures that supersymmetry can't be spontaneously broken by simple variation of the wiring parameter.
\begin{corollary}
On a graph, the vacuum state can be either purely bosonic or contain bosons and fermions, but cannot be purely fermionic. \footnote{This result marks a significant departure from Requardt's results, which claim that for an infinite connected digraph, the zero eigenspace is purely fermionic.  Requardt builds his supersymmetric model from a slightly different graph theoretic perspective than we do, which emphasizes the ingoing and outgoing orientation of edges. However, our results of basic SUSY graph structures in Section \ref{s: SUSY_L_spectrum}  agree with his, so this is somewhat surprising. The discrepancy arises from Requardt's unexplained assertion that for infinite connected graphs, the kernel of $d=I^T$ is $0$, which differs from our own findings that the normalized constant state spans the kernel.}
\end{corollary}

Some types of graphs are more stable under rewiring than others, particularly those which are either very densely or very sparsely connected. For example, the complete graph $K_n$ has the maximum possible number of cycles for a graph of order $n$ and cannot have any bottlenecks. Under a smooth rewiring, only the orientation of the edges of $K_n$ can change. On the other hand, its complement, the empty graph, cannot be rewired at all, since its edge set is empty. However, in both of these cases, by Corollary \ref{cor: o_indep_spectra}, rewiring will leave the spectrum of the SUSY graph Laplacian unchanged. This, however, is not true in most cases. Let's see a few examples how rewiring affects the spectrum of $\Delta_S$. Notice that in each of examples, although the eigenvalues change under rewiring, the spectrum maintains even degeneracy of non-zero eigenvalues. This is unsurprising, because as we rewire, we are simply producing new graphs, and we've shown that the this spectral degeneracy property holds for all graphs (Theorem \ref{thm:SUSY_degeneracy}). 

\begin{example}
The top graph of Figure \ref{fig:Rewiring_4} shows a system with Witten index $1$ along with a cartoon of the spectrum of its SUSY Laplacian. Below it is the resultant path graph (and its spectrum) obtained by rewiring $e_3= v_1v_3$ to $e_3=v_3v_4$. This rewiring raises one bosonic and one fermionic vacuum state together, leaving a purely bosonic vacuum and shifts the excited energy levels. 
    \begin{figure}[ht]
    \centering
    \includegraphics[width=0.4\linewidth]{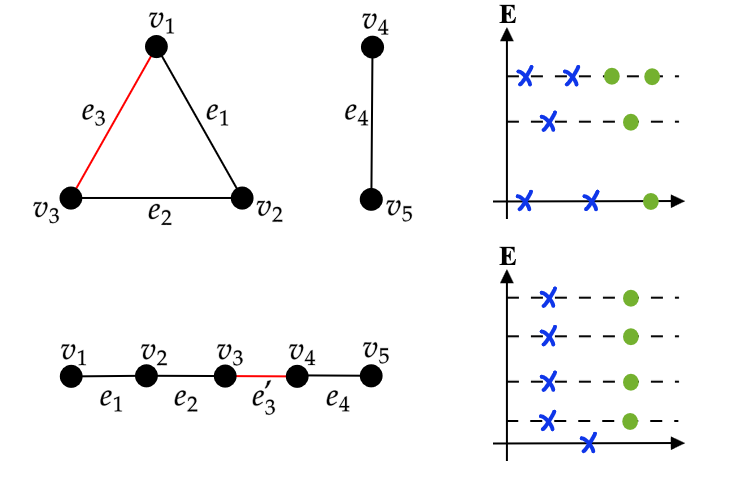}
    \caption{\small{Two graphs are shown at left. Rewiring of the top graph produces the path graph on the bottom. A cartoon of the spectrum of each graph's SUSY Lapalcian is shown at right, with vertex-bosonic states represented by blue x's and edge-fermionic states by green dots.}}
    \label{fig:Rewiring_4}
\end{figure}
\end{example}

\begin{example}
Figure \ref{fig:Rewiring_1} shows a graph of two disconnected 3-cycles and its spectrum. Any choice of single-step rewiring will produce a resultant graph isomorphic (up to orientation) to the rewiring of $e_2$ shown below it, raising one of the two pairs of vacuum states. All such single-step rewirings will result in the same spectrum as is shown at the bottom right, though they may differ in their fermionic steady states. Unlike in the previous example, in this configuration, the Witten index condition for spontaneous supersymmetry breaking holds, though the resultant graph cannot be subsequently rewired to raise the remaining bosonic and fermionic vacuum states above zero. 
\begin{figure}[ht]
    \centering
    \includegraphics[width=0.4\linewidth]{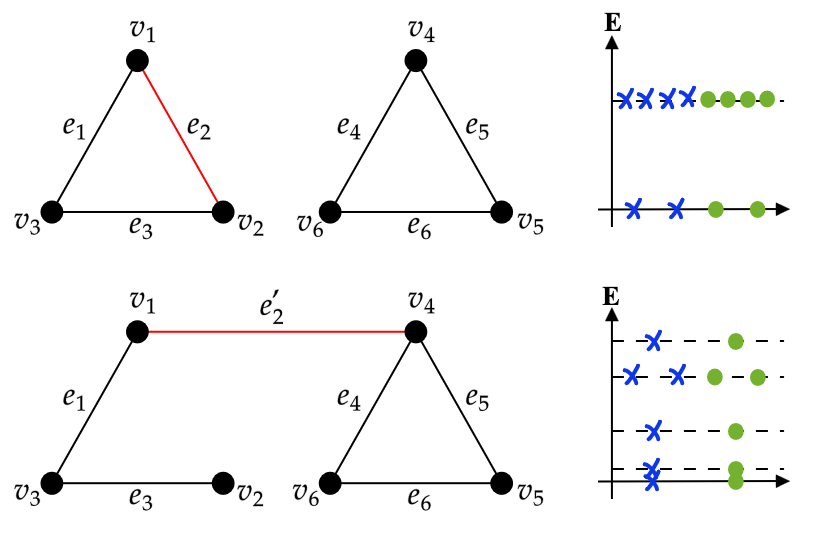}
    \caption{\small{Rewiring a disconnected $3$-cycle to connect the graph raises a bosonic and fermionic vacuum state together, preserving a Witten index of  $0$.}}
    \label{fig:Rewiring_1}
\end{figure}

\end{example}

\begin{example}
In Figure, \ref{fig:Rewiring_6}, the bottleneck edge $e_3$ is rewired to disconnect the graph and form a cycle, increasing the number of vacuum states.
\begin{figure}[ht]
    \centering
    \includegraphics[width=0.45\linewidth]{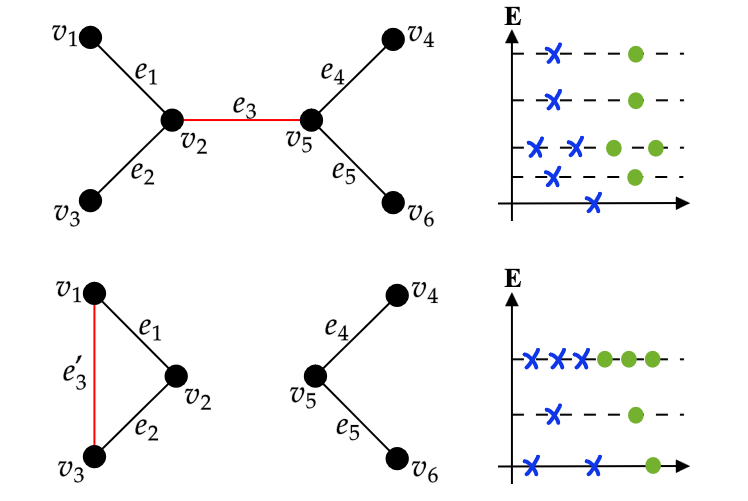}
    \caption{\small{Rewiring a bottleneck to disconnect a graph forms a new cycle and lowers a fermionic and bosonic state to zero-energy.}}
    \label{fig:Rewiring_6}
\end{figure}
\end{example}

\textit{Remark:} As we saw in Section \ref{s:ExitSF_G}, the ability to partition a graph, and therefore the bottlenecks of a graph, are related to the Cheeger inequality for the Fiedler value of the even graph Laplacian on connected graphs. While similar bounds on the second eigenvalue for the odd Laplacian have not, to our knowledge, been explored, we can immediately make two observations based on the equivalence of the positive eigenvalues of the even and odd graph Laplacians. First, as in the case of the top graph in Figure \ref{fig:Rewiring_6}, if the graph is connected, the multiplicity of $0$ in the spectrum of the even Laplacian is $1$, so the first excited energy (and the smallest eigenvalue of the odd Laplacian) is the Fiedler value of the even Laplacian. Second, if $W=0$, then the entire spectra of the two Laplacians are identical, so the second smallest eigenvalue of the odd Laplacian is the same as the Fiedler value of the even Laplacian. 
Further examination of the Fielder value and its analog for $\Delta_-$ could help us understand the spectrum of $\Delta_S$ and how the first excited state changes under rewiring. 

\section{Conclusions and Further Directions} \label{conclusions}

Motivated by the deep connections between supersymmetry and Morse theory and knowledge of a combinatorial version of Morse theory, we have presented a new addition to the graph quantum mechanics model which incorporates supersymmetry and shown that graphs have a naturally supersymmetric structure. Letting vertex states represent bosonic states and edge states represent fermionic states, we found a graph analog of the fermion parity operator to distinguish them and showed that the graph incidence matrix acts like a supercharge transforming between them. We defined a supersymmetric free graph Hamiltonian as the direct sum of the even and odd graph Laplacians and demonstrated that it has non-negative eigenvalues and at least twofold degeneracy in its excited states, as is true for SUSY Hamiltonians in the continuum.  Finally, we recognized the graph Dirac operator as these SUSY graph Hamiltonians, using it  to derive finitely convergent “walk sums,” graph analogs of supersymmetric path integrals. 
 
Based on a combinatorial Morse theory proof that the Euler characteristic of a graph is equal to the Witten index of a graph quantum system, we proved that the $W=0$ condition for spontaneous supersymmetry breaking holds only when $|V|=|E|,$ but that even on such graphs, supersymmetry in the combinatorial, free case cannot be spontaneously broken. Finally, we looked at how rewiring changes the density of the vacuum states of a graph Hilbert space, establishing that zero-energy states are removed and created in vertex-bosonic and edge-fermionic pairs, as they are a continuous SUSY system.  

We anticipate these results to be a starting point for leveraging the finite discrete nature of graphs to model more complicated systems. While our work establishes a model using combinatorial graphs in the free case, in the future we hope to incorporate metrics and potentials, to enable us to discretely model more “realistic” SUSY systems and to examine their continuum limits. We conjecture that in these more complicated cases, rewiring might cause graph Supersymmetry to be spontaneously broken. Further, comparing these results with other discrete supersymmetric models, such as Catterall et al.’s efforts to reconcile lattice QCD with supersymmetry \cite{Catterall_2015}, would be worthwhile, as well as exploring other combinatorial QFTs \cite{contreras2024combinatorial}, for instance, scalar field theory. Using the first quantization formalism here provides a way to represent QFT on a graph by mapping Feynman graphs to the spacetime graph (edges to paths). We  anticipate the utility of graphs as a promising toy model for supersymmetric quantum field theory and opens many new doors for future investigation.

\small{\printbibliography}

\end{document}